\documentclass[twocolumn]{autart}    

\usepackage{graphicx}          
\usepackage{tabularx,booktabs}

\usepackage{comment}
\usepackage{float}
\usepackage{color}
\usepackage{amsmath,amssymb}
\usepackage{amsfonts}
\usepackage{textcomp}
\usepackage{psfrag}
\usepackage{multimedia}
\usepackage{fancybox}
\usepackage{mathtools}
\newtheorem{theorem}{Theorem}

\newtheorem{Lemma}{Lemma}
\newtheorem{Definition}{Definition}
\newtheorem{remark}{Remark}

\newtheorem{Assumption}{Assumption}
\usepackage{enumerate,enumitem}
\newcommand{\myqed}{\hfill $\blacksquare$}
\usepackage{bm} 
\usepackage{tikz}
\DeclareMathOperator{\F}{\rotatebox[origin=c]{45}{\scalebox{0.9}{$\Box$}}}
\usetikzlibrary{shapes.geometric, arrows.meta, positioning}
\usepackage{standalone}
\usepackage[T1]{fontenc}
\definecolor{subsectioncolor}{rgb}{0,0.541,0.855}

\begin{document}

\begin{frontmatter}

\title{ Almost Sure Reachability\\ in Continuous-time Stochastic Systems}

\author[]{Arash Bahari Kordabad,} 
\author[]{Rupak Majumdar,} 
\author[]{Sadegh Soudjani}

\address{Max Planck Institute for Software Systems, Kaiserslautern, Germany (e-mails: {\{arashbk,rupak,sadegh\}@mpi-sws.org}).}


\begin{abstract}
We provide certificates for almost sure reachability of continuous-time stochastic systems governed by stochastic differential equations (SDEs). We first show that a standard Euler–Maruyama discretization may fail to preserve almost sure reachability property of the system using a double-well Langevin system. This observation motivates us to develop certificates for almost sure reachability directly on the continuous-time system. We introduce a pair of certificates, a \emph{drift} function and a \emph{variant} function, and prove necessity and sufficiency for almost sure reachability of an open bounded target set. Using these certificates, for linear SDEs, we give a characterization of almost sure reachability in terms of the spectral structure of the system matrices. For polynomial SDEs, we fix a polynomial template for the drift function and choose the variant function template as an exponential function composed with a polynomial. This allows us to translate the conditions in the certificates into sum-of-squares (SOS) constraints. We then propose an alternating scheme to resolve bilinearities. We illustrate the approach on the double-well Langevin example, showing that continuous-time SOS certificates recover almost sure reachability that is lost under time discretization. Moreover, we verify the SOS approach on a polynomial system.
\end{abstract}
\end{frontmatter}

\section{Introduction}
\label{sec:intro}

Stochastic differential equations (SDEs) play a central role in modeling stochastic dynamical phenomena in physics, chemistry, biology, finance, and control engineering~\cite{oksendal2003stochastic}.  In many such systems, a key qualitative question is whether trajectories \emph{almost surely} reach a desired set despite randomness and initial state~\cite{prandini2006stochastic}.  For instance, in metastable systems modeled by Langevin dynamics, the ability of the process to transition between potential wells with probability one is critical for understanding reaction pathways and long-term behavior~\cite{aristoff2023weighted}. 

Recent development in discrete-time systems has provided two certificates---a \emph{drift} and a \emph{variant} function---that are both necessary and sufficient for almost sure reachability and can be verified using only the one-step transition kernel associated with the system dynamics~\cite{majumdar2024necessary}. 
A natural question is whether one can apply these certificates to a discretized version of the systems modeled with continuous-time SDEs, thereby leveraging discrete-time tools such as sum-of-squares (SOS)~\cite{kordabad2025sum}. 
Unfortunately, we show that an Euler–Maruyama discretization may fail to preserve almost sure reachability of the original continuous-time system.
This is also aligned with available results on numerical instability~\cite{ardourel2021numerical}. 
While more stable discretization schemes can mitigate divergence, they typically do not preserve the polynomial structure needed by, e.g., SOS methods. This demonstrates that almost sure reachability is not preserved under naive discretization, and thus continuous-time certificates must be developed directly for SDEs.

Our contribution is to develop such certificates.

\medskip
\textbf{Related works.}
Classical analyses of stochastic dynamical systems rely on supermartingale arguments to characterize stability, recurrence, and non-evanescence of Markov processes~\cite{meyn2012markov}. In parallel, probabilistic certificates such as stochastic barrier functions have been developed for safety and reachability in continuous-time and hybrid systems~\cite{kordabad2024control,santoyo2021barrier}.  These approaches typically give sufficient conditions for safety or bounded reachability probabilities, but they are not necessary and sufficient criteria for almost sure properties.

Many studies on safety and reachability of continuous-time stochastic systems are based on barrier-type certificates. Stochastic control barrier function (CBF) frameworks extend deterministic CBF ideas to diffusion processes and provide sufficient conditions for forward invariance with probability one~\cite{clark2021control,so2023almost}. Subsequent works refine these conditions and their use in safety-critical control design, for example by proposing alternative stochastic zeroing-type certificates that quantify the probability of staying in a safe set~\cite{nishimura2024control}. Related formulations introduce stochastic control barrier functions that aim to reduce conservatism and control effort compared to reciprocal or zeroing variants, and extend the theory to high-order constraints~\cite{wang2021safety}. In parallel, learning-based representations have been explored through stochastic neural CBFs~\cite{zhang2025stochastic}. Beyond safety filtering, stochastic CBFs have also been integrated into safe optimal control and trajectory optimization pipelines, including approaches based on forward--backward SDEs and deep learning~\cite{pereira2021safe}.

While barrier-type certificates are often presented as sufficient conditions, a parallel line of work studies \emph{converse} statements showing that certificates also exist whenever the system is safe, stable or reachable. In deterministic systems, necessity results for barrier certificates have been established under fairly general settings using convex duality and density-function arguments, showing that safety implies the existence of a barrier certificate ~\cite{prajna2005necessity}. Related converse theorems prove that every robustly safe differential equation admits a barrier certificate and provide explicit constructions based on finite-horizon reachable sets~\cite{ratschan2018converse}. In the same spirit, converse problems have been investigated for specific barrier subclasses used in control yielding necessary and sufficient conditions for the existence of such certificates~\cite{lyu2025converse}.  In stochastic settings, converse drift condition results are classical for Markov processes: for discrete-time Markov chains, \cite[Ch.~9]{meyn2012markov} develops conditions under which reachability implies the existence of a Lyapunov-type drift function.

A related line of work has studied termination and reachability in probabilistic programs using such certificates.  For probabilistic programs, stochastic invariants and ranking supermartingales have been extensively used to reason about both qualitative and quantitative termination~\cite{chatterjee2017stochastic,chatterjee2022sound}. These approaches typically rely on \emph{expected decrease} conditions to ensure progress.  Supermartingale-based certificates have also been extended to more general specifications, including $\omega$-regular objectives and quantitative verification~\cite{abate2024stochastic}, where hierarchical or compositional supermartingale constructions are employed. Similarly, $k$-inductive barrier certificates have been adapted to stochastic systems to provide probabilistic safety and reachability guarantees over unbounded horizons~\cite{anand2022k}. Sound and complete proof systems for probabilistic termination---based on ranking supermartingales and progress conditions---have been developed in both qualitative and quantitative forms~\cite{majumdar2025sound,chatterjee2022sound}.  These works highlight the importance of combining non-divergence guarantees with reachability, a perspective closely echoed in drift–variant certificate frameworks. Subsequent work used this theory for linear systems with additive noise~\cite{kordabad2025certificates} to characterize almost sure reachability in terms of the system matrices and to derive explicit drift–variant certificate forms. For polynomial discrete-time systems, SOS-based methods were developed to synthesize polynomial drift–variant certificates~\cite{kordabad2025sum}. These advances highlight the scope of the drift–variant framework; however, they all apply to discrete-time systems.

\bigskip
\textbf{Contributions.}
This paper develops a continuous-time theory of almost sure reachability together with the required computational tools.  
Our main contributions are:

i) We introduce drift and variant certificates for SDEs and prove that, under mild assumptions on the system, they are both sufficient and necessary for almost sure reachability of a bounded open target set.

ii)  For linear systems with additive noise, we derive a characterization of almost sure reachability based on the system matrices. This extends classical recurrence/transience results to general linear SDEs.

iii) For polynomial systems, we express the inequalities of the certificates as SOS constraints, yielding semidefinite programs that synthesize polynomial drift and variant functions efficiently. An alternating scheme is proposed to resolve bilinear terms in the variant conditions. 

\bigskip
\textbf{Outline.}
Section~\ref{sec:motivational-example} presents a double-well Langevin dynamics example and illustrates how discretization breaks reachability.  Section~\ref{sec:per} provides the preliminaries and background on stochastic differential equations and continuous-time Markov processes. Section~\ref{sec:cts-certificates} introduces continuous-time drift and variant certificates and establishes sufficiency and necessity. In Section~\ref{sec:exam:cont}, we revisit the double-well example and provide the proposed almost sure reachability certificates in continuous-time. Section~\ref{sec:linear-sde} provides the reachability characterization for linear SDEs.  
Section~\ref{sec:sos} develops SOS-based synthesis for polynomial SDEs.  
Section~\ref{sec:simulations} presents numerical examples, and Section~\ref{sec:conclusion} concludes the paper.

\bigskip
\textbf{Notation.}
We write $\mathbb N:=\{0,1,2,\ldots\}$, $\mathbb Q$ for the rational numbers, and $\mathbb R$ for the real numbers.  For a set $S$, we denote by $\overline S$ its closure, by $S^{\mathrm c}$ its complement, and by $\partial S$ its boundary. A set $S \subseteq \mathbb{R}^n$ is called {open} if for every point  $x \in S$, there exists a neighborhood of  $x$ entirely contained in $S$. A set is {closed} if its complement is open. A set $S \subset \mathbb{R}^n$ is {bounded} if there exists an $M\in \mathbb{R}_{\geq 0}$ such that $\|x\| \leq M$ for all $x \in S$. A set is said to be {compact} if it is closed and bounded. $S$ is precompact if $\overline S$ is compact. We denote by $\mathcal B(X)$ the Borel $\sigma$-algebra on $X$, namely the $\sigma$-algebra generated by the open subsets of $X$, and by $\mathcal B(X)$ the set of Borel measurable functions from $X$ to $\mathbb R$. A function $f:X\to\mathbb R$ is Borel measurable if $f^{-1}(B)\in\mathcal B(X)$ for every Borel set $B\subseteq\mathbb R$, and we write $C_b(X)$ for the set of bounded continuous functions on $X$. A function $f:X\to\mathbb R\cup\{\pm\infty\}$ is lower semicontinuous if $\{x\in X:f(x)>a\}$ is open for every $a\in\mathbb R$, and upper semicontinuous if $\{x\in X:f(x)<a\}$ is open for every $a\in\mathbb R$. For a matrix $M$, $\sigma(M)$ denotes its spectrum, $\operatorname{Re}(\lambda)$ the real part of $\lambda\in\mathbb C$, $\alpha(M):=\max\{\operatorname{Re}(\lambda):\lambda\in\sigma(M)\}$ its spectral abscissa. Finally, $\mathbb R[x]$ denotes the ring of real multivariate polynomials in $x$, and $\Sigma[x]$ denotes the cone of sum-of-squares polynomials.

\section{Example: Discretization Breaks Reachability}\label{sec:motivational-example}
In this section, we provide a simple yet revealing example that highlights a fundamental challenge in certifying almost sure reachability for continuous-time stochastic systems using discrete-time techniques.  The overdamped Langevin dynamics in a double-well potential provide a classical model of metastable behavior~\cite{bovier2016metastability}: trajectories fluctuate between two energy wells due to thermal noise, and the process reaches either well with probability one.   We will show that the common forward Euler–Maruyama discretization can introduce artificial numerical instabilities that are entirely absent in the continuous-time SDE.   These instabilities not only break almost sure reachability but also make it impossible for any reachability certificate to hold in the discretized model.

Consider the overdamped Langevin stochastic differential equation associated with the double-well potential $\mathcal{U}(x)=(x-1)^2(x+1)^2$:
\begin{align}
\mathrm{d}x(t)&=-\mathcal{U}'(x(t))\mathrm{d}t+\sigma\mathrm{d}W(t)\nonumber\\ &
=(-4x(t)^3+4x(t))\mathrm{d}t+\sigma\mathrm{d}W(t),\label{eq:cont:Lag}
\end{align}
where $x(t)\in \mathbb{R}$ is the state of the system at time $t\ge 0$, $W(t)$ denotes standard Brownian motion, and $\sigma^2=\tfrac{2}{5}$.
The continuous-time model is polynomial and almost surely reachable between the two wells.
In particular, starting from $x_0=-1$, or any other initial state, the process reaches a small target interval around $+1$ with probability one~\cite{bovier2016metastability}.

A natural question is whether this almost sure reachability property is preserved under standard time discretizations. A forward Euler–Maruyama discretization with step size $\Delta t>0$ yields the following polynomial discrete-time system
\begin{equation*}
x_{n+1}=f_{\text{d}}(x_n,w_n):=x_n+4x_n(1-x_n^2)\Delta t+\sqrt{\tfrac{2\Delta t}{5}}w_n,
\end{equation*}
where $x_n$ is the state at time $n=\Delta t$ and $f_{\text{d}}$ is the system dynamics in discrete-time. Random variables  $w_n\stackrel{\text{i.i.d.}}{\sim}\mathcal{N}(0,1)$ are drawn from a standard Gaussian distribution.
While this discretization preserves the polynomial structure of~\eqref{eq:cont:Lag}, it introduces a numerical instability absent in the continuous-time dynamics.

To understand the source of this discrepancy, consider the deterministic part of the discretized dynamics,
\begin{equation*}
g_{\text{d}}(x)=x+4x(1-x^2)\Delta t=x\bigl(1+4\Delta t(1-x^2)\bigr).
\end{equation*}
For $|x|\geq \sqrt{1 + 3/(4\Delta t)}$, we have
\begin{equation*}
1-x^2\leq \frac{-3}{4\Delta t}\Rightarrow1+4\Delta t(1-x^2)\leq -2,
\end{equation*}
and hence $|1+4\Delta t(1-x^2)|\geq 2$.
Therefore, $|g_{\text{d}}(x)|\geq 2|x|$, implying exponential divergence of the deterministic component for any initial state satisfying  $|x|\geq \sqrt{1 + 3/(4\Delta t)}$.

Thus, Euler-type discretizations may violate reachability, even when the original continuous-time system reaches the target almost surely. More stable schemes, such as implicit or tamed Euler and split-step methods, may preserve reachability~\cite{zhang2009split,brehier2023approximation}. However, they typically do not preserve the polynomial structure of the dynamics. Consequently, polynomial SOS-based synthesis methods are no longer applicable.

This limitation motivates the development of reachability certificates formulated directly for continuous-time stochastic systems.

\section{Preliminaries and Background}\label{sec:per}
In this paper, we consider the state space $X\subseteq\mathbb{R}^n$ to be locally compact metric space. Now, consider the following stochastic differential equation (SDE), 
\begin{equation}\label{eq:sde}
\mathrm{d}x(t)=f(x(t))\mathrm{d}t+g(x(t))\,\mathrm{d}W(t),
\end{equation}
with initial condition $x(0)=x_0\in X$, where $x(t)\in X$ is the state at time $t\geq 0$, and $W(t)\in\mathbb{R}^m$ is an $m$-dimensional standard Brownian motion. Moreover, $f:X\to\mathbb{R}^n$ and $g:X\to\mathbb{R}^{n\times m}$ are given functions. 

The SDE is defined on a complete filtered probability space $(\Omega,\mathcal{F},{\mathcal{F}_t},\mathbb{P})$ satisfying the usual conditions (right-continuity and completeness), where $\mathcal{F}$ is the $\sigma$-algebra of measurable events, and $\mathcal{F}_t$ is the natural filtration generated by the Brownian motion $W(t)$ and the initial state $x(0)$:
\begin{equation*}
\mathcal{F}_t=\sigma\big(x(0),W(s):0\le s\le t\big).
\end{equation*}
All expectations and conditional expectations $\mathbb{E}[\cdot\,|\,\mathcal{F}_s]$ are taken with respect to this filtration.
\smallskip

We define the transition semigroup $P_t: \mathcal{B}(X)\rightarrow \mathcal{B}(X)$, for any $t\geq 0$, associated with the SDE~\eqref{eq:sde} by
\begin{equation}
\label{eq:SemiG}
P_t B(x):=\mathbb{E}[B(x(t))\,|\,x(0)=x],\qquad B : X \to \mathbb{R}.
\end{equation}

\begin{Assumption}[Path Regularity]\label{assum:reg}
Functions $f$ and $g$ in~\eqref{eq:sde} are Borel measurable and locally bounded.
Moreover, for~\eqref{eq:sde}, a (weak or strong) solution exists.
\end{Assumption}

Note that right-continuity of the sample paths holds under
Assumption~\ref{assum:reg}.

\begin{Assumption}[Weak Feller Property]\label{assum:WF}
The transition semigroup $(P_t)_{t\ge0}$ defined in~\eqref{eq:SemiG} maps bounded continuous functions to bounded continuous functions for all $t\geq 0$.
\end{Assumption}
A similar notion of a weak Feller semigroup is adopted e.g., in~\cite{xiang2009measure}. 
For example, the Lipschitz continuity of $f$ and $g$ in~\eqref{eq:sde} ensures the  weak Feller property of the associated transition semigroup (see e.g.,~\cite{kwon2025exit}).

\begin{Definition}[Infinitesimal Generator]\label{def:inf_G}
Let $x(t)$ be a Markov process defined in~\eqref{eq:sde}. The \emph{infinitesimal generator} $\mathcal{A}$ of $x(t)$ is defined by
\begin{equation*}
\mathcal{A}B(x):=\lim_{t\downarrow0}\frac{P_t B(x)-B(x)}{t},
\end{equation*}
where $B: X\rightarrow \mathbb{R}$ is in a set $\mathrm{dom}(\mathcal{A})$ (the domain of the operator $\mathcal{A}$) of functions such that the limit exists at the given $x\in X$.
\end{Definition}

For the SDE~\eqref{eq:sde}, the generator takes the explicit form
\begin{equation}\label{eq:inf:exp}
\mathcal{A}B(x)=f(x)^\top\nabla B(x)+\tfrac{1}{2}\operatorname{tr}\big(g(x)^\top\nabla^2 B(x)g(x)\big),
\end{equation}
whenever $B$ is a twice continuously differentiable function.

The classical infinitesimal generator $\mathcal A$ is defined through the
pointwise limit. For the SDE~\eqref{eq:sde}, the existence of this limit typically requires a regularity assumption on both SDE and $B$, such as twice continuous
differentiability of $B$.
However, many functions of interest in reachability and stability analysis,
such as hitting probabilities constructed
via stopping times, are generally only measurable or continuous and do not
belong to $\mathrm{dom}(\mathcal A)$.
To accommodate such functions while retaining a pointwise notion of drift, as in e.g.,~\cite{meyn1993stability}, we use an \emph{extended} (or \emph{Dynkin}) generator,
defined via martingale properties rather than pointwise semigroup
derivatives.

\begin{Definition}[Extended Infinitesimal Generator]\label{def:extgen_general}
Let $x(t)$ be a Markov process defined in~\eqref{eq:sde}.
A Borel function $B:X\to\mathbb R$ is said to belong to the domain $\mathrm{dom}(\mathcal {A}^{\mathrm e})$ of the
\emph{extended infinitesimal generator} $\mathcal A^{\mathrm e}$ if there exists a Borel function $\psi:X\to\mathbb R$
such that for each initial state $x\in X$ and $t> 0$, $\int_0^t |\psi(x(s))|\,ds < \infty$ almost surely, and the process
\begin{equation*}\label{eq:extgen_general_mg}
M_t^{B} \;:=\; B(x(t)) - B(x) - \int_0^t \psi(x(s))\,\mathrm{d}s,
\end{equation*}
is a right-continuous martingale.
In this case we write $\mathcal A^{\mathrm e} B = \psi$.
\end{Definition}

Note that if $B\in\mathrm{dom}(\mathcal A)$ in the sense of Definition~\ref{def:inf_G},
then by Dynkin's formula $M_t^{B}$ is martingale in Definition~\ref{def:extgen_general}, and $B\in\mathrm{dom}(\mathcal A^{\mathrm e})$ with
$\mathcal A^{\mathrm e}B=\mathcal AB$.

\begin{Definition}[First Hitting Time]\label{def:Hitt}
For a measurable set $A\subseteq X$, the \emph{first hitting time} is defined as
\begin{equation*}
\sigma_A:=\inf\{t\ge0:x(t)\in A\}.
\end{equation*}
\end{Definition}

Throughout, fix an increasing family of open \emph{precompact} sets
$\{O_m:m\in\mathbb Z_+\}$ such that $O_m \uparrow X$ as $m\to\infty$. Such sets exist by Lindel\"of's theorem~\cite{meyn2012markov}.

For each $m\in\mathbb Z_+$, we define the first exit time
from $O_m$, denoted by $T^m$, as the first hitting time of its complement, i.e., $T^m := \sigma_{O_m^c}$,
with the convention that $\inf\emptyset=\infty$.

In order to localize the reachability certificates to precompact sets and accommodate non-smooth test functions, it is standard to introduce truncations of the process, see, e.g.,~\cite{meyn1993stability}. This truncated notation is used primarily in the necessity part of our results in the next section, where we explicitly construct a drift certificate and require a localized pointwise drift characterization.

We next define the truncated
process $x^m(t)$ by
\begin{equation}\label{eq:trunc_process}
x^m(t) \;:=\;
\begin{cases}
x(t), & t < T^m,\\
\Delta_m, & t \ge T^m,
\end{cases}
\end{equation}
for any fixed $\Delta_m\in O_m^c\subset X$. Equivalently, $x^m$ coincides with $x$ on $[0,T^m)$ and is absorbed at
$\Delta_m$ after leaving $O_m$.

Analogous to ~\eqref{eq:SemiG}, let $(P_t^m)_{t\ge 0}$ denote the transition semigroup associated with
$x^m$, i.e., for any $B: X\to\mathbb R$, 
\begin{equation*}
P_t^m B(x)
\;:=\;
\mathbb E\left[B(x^m(t))\,\middle|\,x^m(0)=x\right].
\end{equation*}

\begin{Definition}\label{def:trunc_ext_gen}

{\normalfont\bfseries(Truncated Extended Infinitesimal Generator)}
For a fixed $m\in\mathbb Z_+$, a Borel function $B:X\to\mathbb R$
is said to belong to the domain $\mathrm{dom}(\mathcal A_m^{\mathrm e})$
of the \emph{truncated extended infinitesimal generator}
$\mathcal A_m^{\mathrm e}$ if there exists a Borel function
$\psi:X\to\mathbb R$ such that for each initial state $x\in X$ and
$t\ge 0$, $
\int_0^t |\psi(x^m(s))|\,ds < \infty,
$ and the process
\begin{equation*}
M_t^{B,m}
\;:=\;
B(x^m(t)) - B(x^m(0)) - \int_0^t \psi(x^m(s))\,\mathrm{d}s,
\end{equation*}
is a right-continuous martingale.
In this case we write $\mathcal A_m^{\mathrm e} B = \psi$.
\end{Definition}

The truncation preserves the dynamics on $O_m$ up to $T^m$.
In particular, if $B\in \mathrm{dom}(\mathcal A^{\mathrm e})$,
then  $ B\in \mathrm{dom}(\mathcal A_m^{\mathrm e})$ and $\mathcal A_m^{\mathrm e} B(x)
=
\mathcal A^{\mathrm e} B (x)$ for all $x\in O_m$~\cite{meyn1993stability}.

Consequently, any statements established for $\mathcal A_m^{\mathrm e}$
on $O_m$ can be viewed as localized versions of the corresponding statements
for $\mathcal A^{\mathrm e}$, and global conclusions for the original process
are obtained by letting $m\to\infty$ and the monotonicity of
$T^m$.

\section{Continuous-Time Reachability Certificates}\label{sec:cts-certificates}
In this section, we develop a continuous-time drift--variant certificate framework for almost sure reachability. Almost sure reachability of $G$, in the sense of Definition~\ref{def:Hitt}, means $\mathbb{P}_{x_0}(\sigma_G<\infty)=1$ for all initial states $x_0\in X$.
Our goal is to characterize almost sure reachability for SDEs directly at the continuous-time level, avoiding discretization altogether.  
To this end, we work with the infinitesimal generator of the SDE and introduce two functions—a \emph{drift} and a \emph{variant}—that together play the same conceptual roles as in the discrete-time theory: the drift rules out divergence with positive probability, while the variant establishes probabilistic progress toward the target.  
We will show that these certificates are both sufficient and necessary for almost sure reachability of a bounded open set.

\textbf{V1: Drift certificate.}
A function $V:X\to\mathbb{R}_{\ge 0}$ is called a \emph{drift certificate}
if $V$ is \emph{norm-like}, i.e., $V(x)\to\infty$ as $\|x\|\to\infty$, and there exist
a compact set $C\subset X$ and a constant $d>0$ such that for every $m\in\mathbb Z_+$,
\begin{equation}\label{eq:V1_CD1_match}
\mathcal A_m^{\mathrm e} V(x)\le d\,\mathbf 1_C(x),
\qquad \forall x\in O_m,
\end{equation}
where $\mathbf 1_C:X\to\{0,1\}$ denotes the indicator function of the set $C$, i.e.,
$\mathbf 1_C(x)=1$ if $x\in C$ and $\mathbf 1_C(x)=0$ otherwise.

\textbf{V2: Variant certificate.} 
Given a drift $V$ satisfying \textbf{V1}, a function $U:X\to\mathbb{R}$ is a variant if there exist supporting functions $H:\mathbb{R}_{>0}\to\mathbb{R}$, $\delta:\mathbb{R}_{>0}\to\mathbb{R}_{>0}$, $\epsilon:\mathbb{R}_{>0}\to(0,1]$, and $h:\mathbb{R}_{>0}\to\mathbb{R}_{>0}$, such that for every $r>0$ and for all $x$,
\begin{enumerate}[label=(\roman*), leftmargin=1.4em]
  \item $V(x)\le r\ \Rightarrow\ U(x)\le H(r)$.
  \item For every $x$ with $V(x)\le r$ and $U(x)>0$,
  \begin{equation}\label{eq:ctV2}
   \mathbb{P}_{x}\Big(U\big(x(h(r))\big)-U(x)\ \le\ -\delta(r)\Big)\ \ge\ \epsilon(r).
  \end{equation}
\end{enumerate}

Intuitively, the drift and variant conditions ensure that:  
(i) the process does not escape to infinity by \textbf{V1}, and  
(ii) it makes uniform probabilistic progress toward the target whenever it is outside it by \textbf{V2}. 

\begin{remark}\label{rem:ct_vs_dt_certificates} The certificates in \textbf{V1}--\textbf{V2} are continuous-time analogues of the discrete-time drift/variant criteria proposed in~\cite{majumdar2024necessary}. In discrete-time, drift is stated via the one-step decrement $\Delta V(x)=\mathbb E[V(x_{k+1})\mid x_k=x]-V(x)$ and is required to be non-positive outside a compact set. In continuous-time, the corresponding notion is expressed through the generator; since the classical generator may not apply to non-smooth certificates, we use the truncated extended generator $\mathcal A_m^{\mathrm e}$ and impose the localized uniform inequality \eqref{eq:V1_CD1_match} on each $O_m$. Similarly, the discrete-time variant condition enforces a negative decrement of $U$ in one step with positive probability, while in continuous time there is no distinguished unit step, so we require an analogous decrement over an arbitrary time window of length $h(r)$ as in~\eqref{eq:ctV2}.
\end{remark}

In the following, we establish the sufficiency and necessity theorems for almost sure reachability based on these continuous-time certificates.

\begin{theorem}[Sufficiency]\label{thm:ctsuff}
Under Assumptions~\ref{assum:reg} and~\ref{assum:WF} for the SDE~\eqref{eq:sde}, let $G \subset X$ be an open and precompact set. 
If there exist a drift certificate $V$ satisfying \textbf{V1} and a variant certificate $U$ satisfying \textbf{V2} with $G\supset\{x:U(x)\le 0\}$, then 
\begin{equation*}
\mathbb{P}_{x_0}(\sigma_G<\infty)=1\qquad\text{for all }x_0\in X. 
\end{equation*}
\end{theorem}
\begin{pf}
\noindent\emph{Proof sketch.} We fix $r>0$ and consider the precompact set $C_V(r):=\{x\in X: V(x)\leq r\}$. By \textbf{V2}(i), the values of $U$ on $C_V(r)$ are uniformly bounded, so $C_V(r)$ can be split into finitely many ``levels'' of $U$.
Sampling the process every $h(r)$ units of time yields a discrete-time chain.
Conditioned on staying in $C_V(r)$, \textbf{V2}(ii) guarantees a uniform positive probability of moving to a lower level whenever $U(x)>0$, so the chain cannot stay forever away from $\{x\in X:  U(x)\le 0\}\subseteq G$.
Letting $r\to\infty$ and using norm-likeness of $V$ and continuity of sample paths gives almost sure reachability of $G$.

Under \textbf{V1}, the conclusion of Theorem~3.1(i) in~\cite{meyn1993stability} applies, and since $x(t)$ is a right-continuous process, it follows that the process is non-evanescent, i.e.,
\begin{equation}\label{eq:non_eva}
    \mathbb P_x\left(\lim_{t\to\infty}\|x(t)\|=\infty\right)=0,
\qquad \forall x\in X.
\end{equation}

We now define the following sets inside $C_V(r)$ by levels of $U$:
\begin{align*}
 G_0&\!:=\!\{x\!\in\! C_V(r):  U(x)\le 0\}\subset G,\\
    G_n&\!:=\!\{x\!\in\! C_V(r): (n-1)\delta(r)\! < \!U(x)\!\le\! n\delta(r)\}, n\in \mathbb{N}_{>0}
\end{align*}
and $B_n:=\cup_{i=0}^n G_i\subseteq C_V(r)$ for all $n\in \mathbb{N}$.  By \textbf{V2} (i), for all $x$ satisfying $V(x)\le r$ we have $U(x)\le H(r)$.
Define the integer $M := \lceil \frac{H(r)}{\delta(r)} \rceil$. Then $M$ is the smallest integer such that $M\delta(r) \ge H(r)$.
Then, for every $x\in C_V(r)$, we have $U(x)\le H(r)\le M\,\delta(r),$ which implies $x\in B_{M}$. Hence $C_V(r)= B_{M}$. 

Let $X_k:=x(kh(r))$, $k\in\mathbb{N}\cup\{0\}$ be sample points belonging to the trajectory $x(\cdot)$ and $\Xi:=\{X_k\}_{k\ge 0}$
be the sampled discrete-time trajectory. This forms a discrete-time Markov chain with transition kernel induced by the semigroup~\eqref{eq:SemiG}. From~\eqref{eq:non_eva} one can observe that,
\begin{equation}\label{eq:non_eva:dis}
   \mathbb P_{x_0}\left(\lim_{k\to\infty}\|X_k\|=\infty\right)=0,
\qquad \forall x_0\in X.
\end{equation}
Let $\Box A$ denote the event that a property $A$ holds for all times, and $\F A$ denote that it holds eventually. Then, for all $x\in C_V(r)$ and $U(x)>0$,  \textbf{V2}(ii) yields
\begin{equation} \label{eq:dis:U}
\mathbb P_{x}\left(U(X_{k+1})-U(X_k)\le -\delta(r) \,|\, X_k=x\right)\ge \epsilon(r),
\end{equation}
which matches the discrete-time variant condition in~\cite{majumdar2024necessary}. From the proof of Theorem~1 in~\cite{majumdar2024necessary}, we have,
\begin{equation*}
\mathbb{P}_{x_0}\bigl( \Xi \models \Diamond B_n \,\wedge\, \Box B_M \bigr)
\;=\;
\mathbb{P}_{x_0}\bigl( \Xi \models \Box B_M \bigr).
\end{equation*}
for all $n\in\{0,1,\ldots, M-1\}$. Then, from $B_0\subset G$, we have,
\begin{align*}
   &\mathbb{P}_{x_0}(\sigma_G < \infty)
\;\ge\;
\mathbb{P}_{x_0}(\Xi \models \Diamond G)
\;\ge\;
\mathbb{P}_{x_0}(\Xi \models \Diamond B_0) \\ &
\;\ge\;
\mathbb{P}_{x_0}(\Xi \models \Box\Diamond B_0)  \ge\;
\mathbb{P}_{x_0}(\Xi \models \Box\Diamond B_0 \wedge \Box B_M)
 \\ & \;=\;
\mathbb{P}_{x_0}(\Xi \models \Box B_M)
\;=\;
\mathbb{P}_{x_0}(\Xi \models \Box C_V(r)).
\end{align*}
Note that the first inequality holds because whenever the sampled sequence $\Xi$ eventually reaches $G$, the underlying continuous-time trajectory must hit $G$ at some finite time. From~\eqref{eq:non_eva:dis}, we have $\lim_{r\rightarrow \infty} \mathbb{P}_{x_0}(\Xi \models \Box C_V(r))=1$ for all $x_0$ and hence $\mathbb{P}_{x_0}(\sigma_G < \infty)=1$ for all $x_0\in X$.\myqed
\end{pf}

In the following, we provide two lemmas that will be used to establish the necessity result.

\begin{Lemma}\label{lem:semicont-prob}
Under Assumption~\ref{assum:WF} for the SDE~\eqref{eq:sde}, for every open set $O\subseteq X$, every closed set $F\subseteq X$, and every fixed $t\geq 0$, the maps
\begin{enumerate}
    \item
$
x\mapsto \mathbb P_x(x(t)\in O)\quad \text{and} \quad x\mapsto \mathbb P_x(x(t)\in F),
$
are lower semicontinuous and upper semicontinuous on $X$, respectively.

\item 
$
x\mapsto \mathbb{P}_{x}(\sigma_{O}\le t)\quad \text{and} \quad x\mapsto \mathbb P_x(\sigma_{F}\le t),
$ are lower semicontinuous and upper semicontinuous on $X$, respectively.
\end{enumerate}

\end{Lemma}
\begin{pf}
   (1) Let $O\subseteq X$ be open. Since $X$ is a metric space, the distance-to-the-complement
function $x\mapsto \mathrm{dist}(x,O^{\mathrm c}):=\inf_{y\in O^{\mathrm c}} \mathrm{dist}(x,y)$ is continuous ($1$-Lipschitz). For each $k\in\mathbb N$,
define the bounded continuous function
\begin{equation*}
    \varphi_k(x):=\min\{1,\;k\, \mathrm{dist}(x,O^{\mathrm c})\},\qquad x\in X.
\end{equation*}
Then $0\le \varphi_k\le 1$ and $\varphi_k(x)\uparrow \mathbf 1_O(x)$ pointwise as
$k\to\infty$.  Applying $P_t$ and using monotone convergence,
\begin{equation*}
    P_t\varphi_k (x)=\mathbb E_x[\varphi_k(x(t))]
\uparrow
\mathbb E_x[\mathbf 1_O(x(t))]=\mathbb P_x(x(t)\in O),
\end{equation*}
as $k\to\infty$. By the weak Feller property, each $P_t\varphi_k$ is continuous. Therefore,
$
x\mapsto \mathbb P_x(x(t)\in O)
=
\sup_{k\in\mathbb N} P_t\varphi_k (x)
$
is the pointwise supremum of continuous functions, hence it is lower semicontinuous. For closed set $F\subseteq X$ the proof is obtained similarly by noting that $F^{\mathrm c}$ is open and $\mathbb P_x(x(t)\in F)=1-\mathbb P_x(x(t)\in F^{\mathrm c})$.

(2) For any open $O$, and by continuity of $x (t)$,
\begin{equation}
    \{\sigma_O\le t\}
=
\cap_{m\in\mathbb N}\ \cup_{q\in\mathbb Q\cap[0,t]}
\{x(q)\in O^{(m)}\},\label{eq:hit-open-approx}
\end{equation}
where $O^{(m)}:=\{y\in X:\ \mathrm{dist}(y,O^{\mathrm c})>1/m\}.$ Note that $[0,t]$ is uncountable, while $\mathbb Q\cap[0,t]$ is countable and dense. Hence, using rational times lets us replace an uncountable union over times by a countable one, which we can handle with measurability and semicontinuity arguments. For fixed $(m,q)$, the map
$x\mapsto \mathbb P_x(x(q)\in A^{(m)})$ is lower semicontinuous by the first part. Therefore,
\begin{equation*}
    x\!\mapsto\! \mathbb P_x(\cup_{q\in\mathbb Q\cap[0,t]}\{x(q)\in O^{(m)}\}\!)
\!\!=\!\!\!\!\!\!
\sup_{q\in\mathbb Q\cap[0,t]}\!\!\mathbb P_x(x(q)\in O^{(m)}),
\end{equation*}
is lower semicontinuous because it is countable supremum of lower semicontinuous functions. Taking the
countable intersection in \eqref{eq:hit-open-approx} and using path continuity of~\eqref{eq:sde},
\begin{equation*}
    \mathbb P_x(\sigma_O\le t)
=
\inf_{m\in\mathbb N}\ \sup_{q\in\mathbb Q\cap[0,t]}
\mathbb P_x\big(x(q)\in O^{(m)}\big).
\end{equation*}
Since an infimum of lower semicontinuous functions is lower semicontinuous, $x\mapsto\mathbb P_x(\sigma_A\le t)$
is lower semicontinuous. The result for closed set $F$ is obtained similarly.\myqed
\end{pf}

\begin{Lemma}[Down--Meyn--Tweedie Resolvent]\label{lem:DMT_trunc}
Let $x(t)$ be a right-continuous Markov process and, for each
$m\in\mathbb Z_+$, let $x^m(t)$ denote the truncated process defined
in~\eqref{eq:trunc_process}, with exit time $T^m=\sigma_{O_m^c}$.
Assume that for every initial state $x\in X$ and every $T>0$,
\begin{equation*}
\mathbb P_x\left(\sup_{0\le t\le T}\|x^m(t)\|<\infty\right)=1 .
\end{equation*}
Let the truncated extended generator $\mathcal A_m^{\mathrm e}$ be defined
as in Definition~\ref{def:trunc_ext_gen}.
Fix $\eta>0$ and define the truncated resolvent, for bounded measurable
$B:X\to\mathbb R$, by
\begin{equation*}
U_\eta^m B(x)
\;:=\;
\mathbb E_x\left[\int_0^\infty \mathrm{e}^{-\eta t} B(x^m(t))\,\mathrm{d}t\right].
\end{equation*}
Then for every bounded measurable $B$, $U_\eta^m B\in
\mathrm{dom}(\mathcal A_m^{\mathrm e})$ and
\begin{equation*}
\mathcal A_m^{\mathrm e} \, U_\eta^m B
\;=\;
- B + \eta U_\eta^m B,
\qquad \forall x\in O_m .
\end{equation*}
\end{Lemma}

\begin{pf}
See \cite[Lemma 4.3(a) and eq.\ (28)]{DMT95}.\myqed
\end{pf}

\begin{theorem}[Necessity]\label{thm:ctnec}
Under Assumptions~\ref{assum:reg} and~\ref{assum:WF} for the SDE~\eqref{eq:sde}, let $G \subset X$ be an open and precompact set. If
$\mathbb{P}_{x_0}(\sigma_G<\infty)=1$ for all $x_0\in X$, then there exist 
\begin{itemize}[leftmargin=1.2em]
\item a drift $V:X\to\mathbb{R}_{\ge 0}$ satisfying  (\textbf{V1}), and
\item a variant $U:X\to\mathbb{R}$ with supporting functions $H,\delta,\epsilon,h$ that satisfies \textbf{V2} and $G\supset\{x:U(x)\le 0\}$.
\end{itemize}
\end{theorem}

\begin{pf}
\noindent\emph{Proof sketch.}
The proof has two parts, constructing a drift certificate $V$ satisfying \textbf{V1} and a variant certificate $U$ satisfying \textbf{V2}.
First, we choose an open precompact neighborhood $O_0$ of the target $G$ and fix a countable increasing precompact exhaustion $O_0\subset O_1\subset\cdots$ of $X$.
Using first-hitting probabilities
$W_n(x)=\mathbb P_x(\sigma_{D_n}<\sigma_{O_0})$ with $D_n=O_n^{\mathrm c}$, we obtain bounded functions that equal $1$ on $D_n$. We then pass from probabilities to drift certificates by applying the truncated resolvent, which produces bounded functions $V_n$ satisfying the resolvent identity, and hence yielding a non-positive truncated drift outside $O_0$ (on each $O_n$). A subsequence $(m_i)$ is then selected so that $V_{m_i}$ becomes uniformly small on each $O_i$, and we define
$V(x)=\sum_{i\ge 0}V_{m_i}(x)$.
This series converges uniformly on each precompact set (ensuring $V\in\mathrm{dom}(\mathcal A_m^{\mathrm e})$ for all $m$), while still diverging at infinity (hence $V$ is norm-like); this yields \textbf{V1}.
Second, we construct $U$ by building outward ``layers'' $\mathcal G_n$ starting from $G$:
each new layer collects states that reach the previous layer within a fixed time $\tau$ with probability exceeding a prescribed threshold.
The weak Feller property ensures that these layers are open, and almost sure reachability implies that their union covers $X$.
Defining $U$ as the layer index makes $\{x\in X: U(x)\le 0\}=G$, bounds $U$ on each $V$-sublevel set, and guarantees a uniform positive probability of decreasing $U$ by at least one over time $\tau$.
With suitable supporting functions $H,\delta,\epsilon,h$, this gives \textbf{V2} and completes the proof.

Pick an open set $O_0$ such that $G\subset O_0$ and $\overline{O_0}$ is compact.
Fix an increasing exhaustion by open precompact sets
\begin{equation*}
O_0\subset O_1\subset O_2\subset \cdots,\quad
\overline{O_i}\subset O_{i+1},\quad
\bigcup_{i\ge 0}O_i=X.
\end{equation*}
By Lindel\"of's theorem~\cite{kelley2017general} (since $X$ is a separable metric space), such sets exist. For each $n\ge 1$, define $D_n:=O_n^{\mathrm c}$ and
\begin{equation*}
W_n(x):=\mathbb P_x\!\left(\sigma_{D_n}<\sigma_{O_0}\right),\quad \forall x\in X.
\end{equation*}
Then $0\le W_n\le 1$, $W_n(x)=1$ for all $x\in D_n$, and $W_n(x)=0$ for all $x\in O_0$. For each $m\ge 1$, let ${\mathcal A}_n^{\mathrm e}$ denote the extended generator of $x^n(t):=x(t\wedge \sigma_{D_n})$.

If $x\in D_n$, then $W_n(x)=1$, and therefore
\begin{equation*}
\mathbb E_x\!\left[W_n\!\left(x^n(t)\right)\right]\le 1=W_n(x).
\end{equation*}
For $x\in O_n\setminus O_0$, by the strong Markov property,
\begin{align*}
W_n(x)
&=\mathbb P_x\!\left(\sigma_{D_n}<\sigma_{O_0}\right) \nonumber\\
&=\mathbb E_x\!\left[\mathbb P_{x^n(t)}\!\left(\sigma_{D_n}<\sigma_{O_0}\right)\right] =\mathbb E_x\!\left[W_n\!\left(x^n(t)\right)\right].
\end{align*}
Hence, for all $x\in X\setminus O_0$ and all $t\ge 0$,
\begin{equation*}
\mathbb E_x\!\left[W_n\!\left(x^n(t)\right)\right]\le W_n(x).
\end{equation*}
Multiplying by $e^{-\eta t}$, integrating over $t\in[0,\infty)$, and interchanging expectation and integration, yields
\begin{equation}\label{eq:Wn_resolvent_bound}
\mathbb E_x\!\left[\int_0^\infty e^{-\eta t}W_n\!\left(x^n(t)\right)\,dt\right]
\le \frac{1}{\eta}\,W_n(x), \, \forall x\in X\setminus O_0.
\end{equation}
Fix $\eta>0$ and define, for each $n\ge 1$,
\begin{equation*}
 V_n(x):=\mathbb E_x\!\left[\int_0^\infty e^{-\eta t}W_n\!\left(x^n(t)\right)\,dt\right]\le 1/\eta,
\end{equation*}
and by \eqref{eq:Wn_resolvent_bound},
\begin{equation}\label{eq:Vn_small_vs_Wn}
0\le  V_n(x)\le \frac{1}{\eta}\,W_n(x),\quad \forall x\in X\setminus O_0.
\end{equation}
By the truncated resolvent identity in Lemma~\ref{lem:DMT_trunc}, $ V_n\in \mathrm{dom}({\mathcal A}_m^{\mathrm e})$ for any $n\geq m$, and
\begin{equation}\label{eq:drift_each_term}
{\mathcal A}_m^{\mathrm e} V_n(x)=-W_n(x)+\eta  V_n(x),\quad  \forall  x\in O_m,
\end{equation}
because $O_m\subset O_n$. Hence, by~\eqref{eq:Vn_small_vs_Wn}, ${\mathcal A}_m^{\mathrm e} V_n(x)\leq 0$ for all $x\in O_m\setminus O_0$.

By Lemma~\ref{lem:semicont-prob}, the map
$x\mapsto \mathbb P_x(\sigma_{O_0}>t)$ is upper semicontinuous for each $t>0$.
Moreover, because $\mathbb P_x(\sigma_{O_0}<\infty)=1$ for all $x$, we have
$\mathbb P_x(\sigma_{O_0}>t)\downarrow 0$ as $t\to\infty$ for each fixed $x$.
By a standard compactness argument on $\overline{O_i}$, there exists a deterministic time $t_i>0$ such that
\begin{equation*}
\sup_{x\in O_i}\mathbb P_x(\sigma_{O_0}>t_i)<2^{-(i+1)}.
\end{equation*}
Since $O_n\uparrow X$ and the sample paths are continuous, for each fixed $x\in X$,
\begin{equation*}
\mathbb P_x(\sigma_{D_n}<t_i)\longrightarrow 0
\qquad\text{as }n\to\infty.
\end{equation*}
Using upper semicontinuity in $x$ and compactness of $\overline{O_i}$ once again, we may choose recursively
$m_i>\max\{i,m_{i-1}\}$, with the convention $m_{-1}:=0$, such that
\begin{equation*}
\sup_{x\in O_i}\mathbb P_x(\sigma_{D_{m_i}}<t_i)<2^{-(i+1)}.
\end{equation*}
Therefore, for every $x\in O_i$,
\begin{align*}
W_{m_i}(x)
&=\mathbb P_x\!\left(\sigma_{D_{m_i}}<\sigma_{O_0}\right) \nonumber\\
&\le \mathbb P_x\!\left(\sigma_{D_{m_i}}<t_i\right)+\mathbb P_x\!\left(\sigma_{O_0}>t_i\right)<2^{-i}. 
\end{align*}
Together with \eqref{eq:Vn_small_vs_Wn}, this gives
\begin{equation*}
\sup_{x\in O_i}V_{m_i}(x)\le \frac{1}{\eta}\sup_{x\in O_i}W_{m_i}(x)
<\frac{2^{-i}}{\eta}.
\end{equation*}
Hence, for each fixed $m\ge 0$,
\begin{align*}
\sum_{i=m}^\infty \sup_{x\in O_m}V_{m_i}(x)
&\le \sum_{i=m}^\infty \sup_{x\in O_i}V_{m_i}(x) \nonumber\\
&< \frac{1}{\eta}\sum_{i=m}^\infty 2^{-i}
<\infty.
\end{align*}
Thus the series $\sum_{i=0}^\infty V_{m_i}$ converges absolutely and uniformly on every $O_m$. Define
\begin{equation*}
V(x):=\sum_{i=0}^\infty V_{m_i}(x),\qquad \forall x\in X.
\end{equation*}
For each $N\ge 0$, let
\begin{equation*}
V^{(N)}\!(x)\!\!:=\!\!\sum_{i=0}^N\! V_{m_i}\!(x),
\,
g^{(N)}\!(x)\!\!:=\!\!\sum_{i=0}^N\bigl(\!-W_{m_i}(x)+\eta V_{m_i}\!(x)\!\bigr)\!.
\end{equation*}
By \eqref{eq:drift_each_term}, for every fixed $m\ge 1$,
\begin{equation*}
V^{(N)}\in \mathrm{dom}(\mathcal A_m^{\mathrm e})\quad\text{and}\quad
\mathcal A_m^{\mathrm e}V^{(N)}(x)=g^{(N)}(x),
\end{equation*}
for all $x\in O_m$. Moreover, for $i\ge m$ and $x\in O_m$, since $O_m\subset O_i$, we have
\begin{align*}
\bigl|-W_{m_i}(x)+\eta V_{m_i}(x)\bigr|
&\le W_{m_i}(x)+\eta V_{m_i}(x) \nonumber\\
&\le 2\,W_{m_i}(x),
\end{align*}
and therefore
\begin{equation*}
\sup_{x\in O_m}\bigl|-W_{m_i}(x)+\eta V_{m_i}(x)\bigr|
\le 2\sup_{x\in O_i}W_{m_i}(x)
<2^{-i+1}.
\end{equation*}
Hence $\sum_{i=0}^\infty g_i$ converges uniformly on every $O_m$, where
\begin{equation*}
g_i(x):=-W_{m_i}(x)+\eta V_{m_i}(x),
\qquad
g(x):=\sum_{i=0}^\infty g_i(x).
\end{equation*}
Since $V^{(N)}\to V$ uniformly on $O_m$ and $g^{(N)}\to g$ uniformly on $O_m$,
the martingale characterization of $\mathcal A_m^{\mathrm e}$ and dominated convergence imply that
\begin{equation*}
V\in \mathrm{dom}(\mathcal A_m^{\mathrm e})\quad\text{and}\quad
\mathcal A_m^{\mathrm e}V(x)=g(x),\qquad \forall x\in O_m.
\end{equation*}
Now set $C:=\overline{O_0}$, which is compact. If $x\in O_m\setminus C$, then $x\notin O_0$, and hence by \eqref{eq:drift_each_term},
\begin{equation*}
g_i(x)=\mathcal A_m^{\mathrm e}V_{m_i}(x)\le 0,\qquad \forall i\ge 0.
\end{equation*}
Therefore
\begin{equation*}
\mathcal A_m^{\mathrm e}V(x)=g(x)\le 0,\quad \forall x\in O_m\setminus C.
\end{equation*}
On the compact set $C$, the term $g_0$ is bounded, and for $i\ge 1$ we have $C\subset O_1\subset O_i$, so
\begin{equation*}
\sup_{x\in C}|g_i(x)|
\le 2\sup_{x\in O_i}W_{m_i}(x)
<2^{-i+1}.
\end{equation*}
Thus $\sum_{i=0}^\infty g_i$ converges uniformly on $C$, and in particular
\begin{equation*}
d:=\sup_{x\in C}g(x)<\infty.
\end{equation*}
Consequently,
\begin{equation*}
\mathcal A_m^{\mathrm e}V(x)=g(x)\le d\,\mathbf 1_C(x),\quad \forall x\in O_m,\, m\ge 1.
\end{equation*}
It remains to prove that $V$ is norm-like. If $x\in D_{m_i}$, since $W_{m_i}(x)=1$ on $D_{m_i}$, we obtain
\begin{equation*}
V_{m_i}(x)=\int_0^\infty e^{-\eta t}\,dt=\frac{1}{\eta},
\quad x\in D_{m_i}.
\end{equation*}
Now fix $M\in \mathbb N$ and set $N_M:=m_M$. If $x\in D_{N_M}$, then $m_i\le N_M$ for every $i\le M$, hence
$O_{m_i}\subset O_{N_M}$ and therefore $D_{N_M}\subset D_{m_i}$.
Thus
\begin{equation*}
V_{m_i}(x)=\frac{1}{\eta},\qquad i=0,1,\dots,M,
\end{equation*}
and so
\begin{equation*}
V(x)\ge \sum_{i=0}^M V_{m_i}(x)=\frac{M+1}{\eta},
\quad  \forall  x\in D_{N_M}=X\setminus O_{N_M}.
\end{equation*}
Since $(O_{N_M})_{M\ge 0}$ is still an increasing precompact exhaustion of $X$, every sublevel set of $V$
is contained in some $\overline{O_{N_M}}$ and is therefore precompact. Hence $V$ is norm-like.
In particular, if $X$ is a normed space, then $V(x)\to\infty$ as $\|x\|\to\infty$. Thus $V$ satisfies \textbf{V1}.

For constructing a variant $U$, fix a small sampling horizon $\tau>0$ and define the sequence of open sets $\mathcal{G}_0:=G$ and 
\begin{equation}\label{eq:Gn-def-sde}
\mathcal{G}_{n+1}
:=\mathcal{G}_n
\cup
\big\{x\in X:\mathbb{P}_x\big(x(\tau)\in\mathcal{G}_n\big)>2^{-(n+1)}\big\},
\end{equation}
for all $n\in \mathbb{N}$. Note that the set
$\{x:\mathbb{P}_x\big(x(\tau)\in\mathcal{G}_n\big)>2^{-(n+1)}\}$ is open because the SDE~\eqref{eq:sde} is weak Feller; hence all $\mathcal{G}_n$ are open. The sets $\mathcal{G}_n$ give layered decomposition of the entire state space and will be used to construct the levels of variant $U$.

We now show that $\cup_{n=0}^{\infty}\mathcal{G}_n = X$. To this end, we define the following increasing sequence of sets
\begin{equation*}
B_n := \big\{x\in X : \mathbb{P}_x(\sigma_G \le n\tau) > 0\big\}, 
\qquad n\in\mathbb{N}.    
\end{equation*}
Clearly $B_0 = G$.  
Since $\mathbb{P}_x(\sigma_G < \infty)=1$, for every $x\in X$ there exists $n\in\mathbb{N}$ such that $\mathbb{P}_x(\sigma_G \leq  n\tau)>0$ and thus $x\in B_n$.  
Hence $X\subseteq \cup_{n=0}^{\infty}B_n$. Noting that $B_n\subseteq X$ for all $n\in \mathbb{N}$, we can conclude  $\cup_{n=0}^{\infty}B_n = X$. Figure~\ref{fig:-1} provides an illustration of the layered sets $B_n$ and $\mathcal{G}_n$, showing how they expand outward from the target $G$ and how trajectories may enter successive layers over time.

\begin{figure}[t]
\centering
\includegraphics[width=0.48
\textwidth]{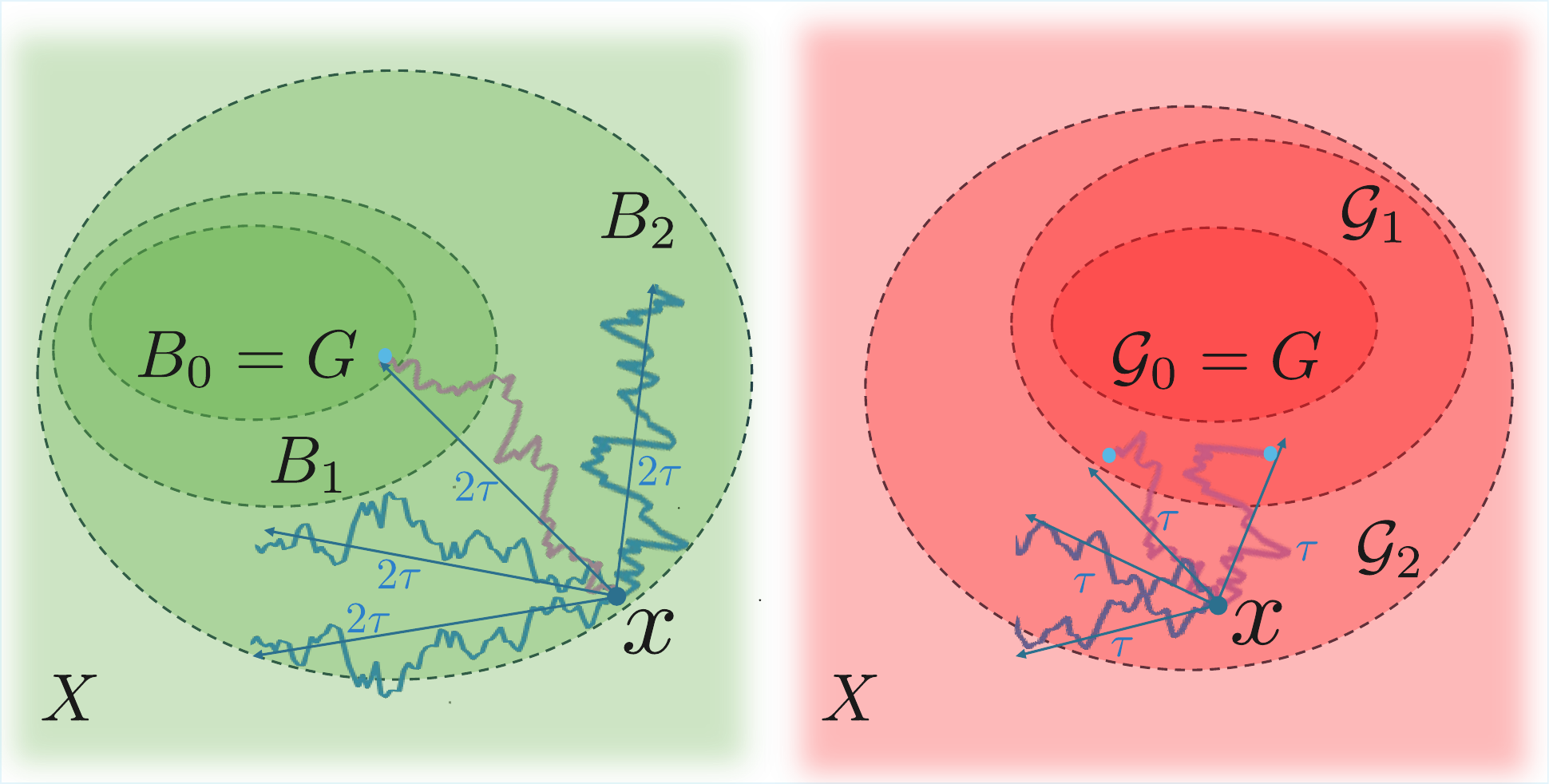}
  \caption{Illustration of the layered sets $B_n$ (left) and $\mathcal{G}_n$ (right) used in the proof.   Left: $B_n$ collects all states from which the process can reach the target set $G$ within time $n\tau$ with positive probability.  Right: $\mathcal{G}_n$ includes $\mathcal{G}_{n-1}$ and all states that reach $\mathcal{G}_{n-1}$ within time $\tau$ with probability exceeding $2^{-n}$.   In both constructions, the layers expand outward until they cover the entire state space.}
  \label{fig:-1}
\end{figure}

We will show by induction on $k$ that
$
B_k \subseteq \cup_{n=0}^{\infty}\mathcal{G}_n
$.
The base case $k=0$ holds because $B_0 = G = \mathcal{G}_0$. Take any $x\in B_{1}\setminus B_0$ but $x\notin \cup_{n=0}^{\infty}\mathcal{G}_n$. The first implies $\mathbb{P}_x(x(\tau)\in B_0)>0$. For the second, from~\eqref{eq:Gn-def-sde} one can observe that,
\begin{equation*}
    x\notin \mathcal{G}_{n+1} \Leftrightarrow \mathbb{P}_x(x(\tau)\in \mathcal{G}_{n})\leq 2^{-(n+1)},
\end{equation*}
for all $n$. Then we have
\begin{align}\label{eq:cont}
  \mathbb{P}_x(x(\tau)\in \mathcal{G}_m)&=   \mathbb{P}_x(x(\tau)\in G)
\\
&+ 
\sum_{n=0}^{m} 
\mathbb{P}_x\big(x(\tau)\in\mathcal{G}_{n+1}\setminus\mathcal{G}_n\big) \leq  2^{-(m+1)}\nonumber
\end{align}
for all $m\in \mathbb{N}$ because $x\notin \cup_{n=0}^{\infty}\mathcal{G}_n$. Hence $\mathbb{P}_x(x(\tau)\in G)\leq 2^{-(m+1)}$ for all $m$ enforcing $\mathbb{P}_x(x(\tau)\in G)= 0$ which is a contradiction. 

Assume that $B_k \subseteq \cup_{n=0}^{\infty}\mathcal{G}_n$.  
Take any $x\in B_{k+1}\setminus B_k$ but $x\notin \cup_{n=0}^{\infty}\mathcal{G}_n$. These imply $\mathbb{P}_x(x(\tau)\in B_0)=0$, $\mathbb{P}_x(x(\tau)\in B_k)>0$, and 
\begin{align*}
  0<&\mathbb{P}_x(x(\tau)\!\in \!B_k)\!\leq  \!\mathbb{P}_x(x(\tau)\!\in\! \cup_{n=0}^{\infty}\mathcal{G}_n)\!=\!\mathbb{P}_x(x(\tau)\!\in \!G)
+ \\  &
\sum_{n=0}^{\infty} 
\mathbb{P}_x\big(x(\tau)\!\in\!\mathcal{G}_{n+1}\!\setminus\!\mathcal{G}_n\big)
 \!=\! 
\sum_{n=0}^{\infty} 
\mathbb{P}_x\big(x(\tau)\!\in\!\mathcal{G}_{n+1}\!\setminus\!\mathcal{G}_n\big). 
\end{align*}
Then there exists some $n$ such that $0<\mathbb{P}_x\big(x(\tau)\in\mathcal{G}_{n+1}\setminus\mathcal{G}_n\big)$. Similar contradiction to~\eqref{eq:cont}, we conclude $B_{k+1}\in \cup_{n=0}^{\infty}\mathcal{G}_n$.
Combining with $\cup_{k=0}^{\infty}B_k = X$ yields
$
{\cup_{n=0}^{\infty}\mathcal{G}_n = X.}
$

Define the function \(U : X \to \mathbb{R}\) by
\begin{equation} \label{eq:U-def-sde}
    U(x) =
\begin{cases}
0, & \text{if } x \in G, \\[6pt]
n+1, & \text{if } x \in \mathcal{G}_{n+1} \setminus \mathcal{G}_n,\quad n \in \mathbb{N}.
\end{cases}
\end{equation}
Let $C_V(r):=\{x:V(x)\le r\}$.
Almost sure reachability and compactness of $C_V(r)$ imply that there exists
an integer $M(r)$ such that $C_V(r)\subseteq\mathcal{G}_{M(r)}$.
Hence $U$ is bounded on the sublevel set $H(r):=M(r)$.

For $n\ge1$ and $x\in\mathcal{G}_n\setminus\mathcal{G}_{n-1}$,
the definition \eqref{eq:Gn-def-sde} gives
\begin{equation*}
    \mathbb{P}_x\big(x(\tau)\in\mathcal{G}_{n-1}\big)>2^{-n}.
\end{equation*}
Because $U(x)=n$ on $\mathcal{G}_n\setminus\mathcal{G}_{n-1}$ and
$U(y)\le n-1$ iff $y\in\mathcal{G}_{n-1}$,
we obtain,
\begin{equation}\label{eq:prob-drop-sde}
\mathbb{P}_x\big(U(x(\tau))-U(x)\le-1\big)
=\mathbb{P}_x\big(x(\tau)\in\mathcal{G}_{n-1}\big)
\ge 2^{-n}.
\end{equation}

We use the fixed sampling horizon $\tau$ as the time step in the variant condition $h(r):=\tau $ and $\delta(r):=1$. Within each sublevel set $C_V(r)$, define
\begin{align*}
    &\epsilon(r)
:=\\ &\min_{1\le n\le M(r)}
\inf_{x\in(\mathcal{G}_n\setminus\mathcal{G}_{n-1})\cap C_V(r)}
\mathbb{P}_x\big(x(\tau)\in\mathcal{G}_{n-1}\big)
>2^{-M(r)}.
\end{align*}
Then from \eqref{eq:prob-drop-sde},
for every $x\in C_V(r)$ with $U(x)>0$,
\begin{align*}
    &\mathbb{P}_x\big(U(x(h(r)))-U(x)\le-\delta(r)\big)
=\\ &\mathbb{P}_x\big(U(x(\tau))-U(x)\le-1\big)
\ge\epsilon(r).
\end{align*}
With
\begin{equation*}
    H(r)=M(r),\quad \delta(r)=1,\quad \epsilon(r)>0,\quad h(r)=\tau,
\end{equation*}
the function $U$ in \eqref{eq:U-def-sde} satisfies the continuous-time variant conditions in~\textbf{V2}.\myqed
\end{pf}

In this section, we introduced drift and variant certificates for continuous-time stochastic systems and established that they fully characterize almost sure reachability.  
The sufficiency theorem showed that the existence of such certificates guarantees that the process reaches the target set with probability one, while the necessity theorem demonstrated that, whenever almost sure reachability holds, appropriate drift and variant functions can always be constructed.  
Together, these results provide a complete  framework and furnish the theoretical foundation for analyzing reachability directly in continuous time, without resorting to discretization.  
In the next section, we revisit the double-well example and use the proposed certificates to explicitly construct continuous-time drift and variant functions.

\section{Example Revisited: Constructing Certificates for the Double-Well Langevin system}\label{sec:exam:cont}
We now return to the double-well Langevin system from Section~\ref{sec:motivational-example} and show how the continuous-time drift and variant certificates developed in Section~\ref{sec:cts-certificates} allow us to certify almost sure reachability. We define the target set as a small neighborhood of the right well, $G=(1-2\rho, 1+2\rho)$, for some small $\rho>0$.

\medskip
\noindent\textbf{Drift certificate.}
We first construct a drift function $V$ for the SDE.  A natural choice for this one-dimensional system is the quadratic function $V(x)=x^2$, which is norm-like and bounded in any compact set. Since this function is sufficiently smooth, we compute its infinitesimal generator as in~\eqref{eq:inf:exp},
\begin{align*}
    \mathcal {A} V(x) &= (-4x^3 + 4x)(2x) + (1/2)2\sigma^2 
         \\ &= -8 x^4 + 8 x^2 + 0.4.
\end{align*}
One can easily see that $\mathcal {A} V(x)\leq 0$ for all $|x|\geq 2$. Hence $V$ satisfies condition \textbf{V1}.

\medskip
\noindent\textbf{Variant certificate.}
We now construct a variant function that ensures probabilistic progress toward the target set as in \textbf{V2}.  For a design parameter $\lambda>0$, we define a variant candidate as follows
\begin{equation*}
U(x):=\frac{1}{\lambda}\big(1-\mathrm{e}^{-\lambda \zeta(x)}\big),\quad \zeta(x):=(x-1)^2-\rho^2.   
\end{equation*}
Let $C_V(r):=\{x\in \mathbb{R}: V(x)\le r\}=\{x\in \mathbb{R}: |x|\le \sqrt r\}$. Then
\begin{equation*}
\zeta(x) \le (\lvert x\rvert+1)^2-\rho^2
 \le (\sqrt r+1)^2-\rho^2=: \zeta_{\max}(r),    
\end{equation*}
and, because $s\mapsto\lambda^{-1}(1-\mathrm{e}^{-\lambda s})$ is increasing,
\begin{equation*}
U(x) \le\ H(r):=\frac{1}{\lambda}\Big(1-\mathrm{e}^{-\lambda \zeta_{\max}(r)}\Big),
\quad \forall x\in C_V(r).    
\end{equation*}
Thus \textbf{V2}(i) holds with this explicit $H(r)$. We verify \textbf{V2}(ii) for all $x\in K_r:=C_V(r) \cap \{x\in\mathbb{R} : U(x)>0\}$.

The target inclusion $\{x\,:\, U(x)\leq 0\}\subset G$ holds by construction. 

For any small $h(r)=\tau>0$, and from It\^o’s formula,
\begin{align*}
Y:=&\,U(x(\tau))-U(x)\\ =&\int_0^\tau \mathcal{A}U(x(s))\,\mathrm{d}s
+\sigma \int_0^\tau U'(x(s))\,\mathrm{d}W(s).   
\end{align*}
We know that
\begin{equation*}
\mathbb{E}[Y]
=\mathbb{E}\int_0^\tau \mathcal{A}U(x(s))\,\mathrm{d}s,
\end{equation*}
where 
\begin{align*}
\mathcal{A}U(x)
&= f(x)\,U'(x)+\tfrac12\sigma^2\,U''(x) \notag\\
&= \mathrm{e}^{-\lambda \zeta(x)}\Big(f(x)\,\zeta'(x)+\tfrac12\sigma^2\zeta''(x)
-\tfrac12\sigma^2\lambda \big(\zeta'(x)\big)^2\Big) \notag\\
&= \!\mathrm{e}^{\!-\!\lambda \zeta(x)}\underbrace{\Big(-8x(x-1)^2(x+1)\!+\!\tfrac{2}{5}\!-\!\tfrac{4}{5}\lambda (x\!-\!1)^2\Big)}_{=:\beta(x)}.
\end{align*}
Since $x(x+1) \ge -\tfrac{1}{4}$ and $\rho^2 \leq (x-1)^2$ on $K_r$, for any $\lambda>\tfrac{5}{2}$, we have,
\begin{equation*}
    \beta(x)\leq \frac{2}{5} - \rho^2\Big(\frac{4}{5}\lambda + 8(-\frac{1}{4})\Big)
       = \frac{2}{5} - \rho^2\Big(\frac{4}{5}\lambda - 2\Big).
\end{equation*}
Therefore, if we pick $\lambda$ such that
$
\lambda\ >\frac{5}{2}+\frac{1}{2\rho^2}\,,
$
then we get $\beta(x)<0$. Hence, for all $x\in K_r$,
\begin{equation*}
\mathcal{A}U(x) \le \mathrm{e}^{-\lambda \zeta(x)}\beta(x) \le \mathrm{e}^{-\lambda \zeta_{\max}(r)}\beta_r
=: -\mu(r) < 0,   
\end{equation*}
with $\mu(r)>0$, where $\beta_r:=\max_{x\in K_r} \beta(x)<0$. Consequently, for any small $h(r)=\tau>0$, $\mathbb{E}[Y]
\le -\mu(r)\tau$, for all $x\in K_r$.
Moreover,
\begin{equation*}
    \mathrm{Var}(Y)
=\mathbb{E}\int_0^\tau \sigma^2 U'(x(s))^2 \mathrm{d}s.
\end{equation*}
On $K_r$, $\zeta(x)\ge 0$ and thus
\begin{equation*}
    \lvert U'(x)\rvert=\lvert \mathrm{e}^{-\lambda \zeta(x)}\,2(x-1)\rvert
 \le 2 (\sqrt r+1)
=: L_r.
\end{equation*}
Hence
$
\mathrm{Var}(Y)
 \le \sigma^2 L_r^2\tau
 = \frac{2}{5}L_r^2\tau,$
 for all $x\in K_r$.
Now choose $\delta(r):=\tfrac{1}{2}\mu(r)\tau$. Cantelli’s one–sided inequality~\cite{ghosh2002probability} gives
\begin{align*}
    \mathbb{P}\big(Y\le -\delta(r)\big)
 &\ge
\frac{\big(\mathbb{E}[-Y]-\delta(r)\big)^2}{\mathrm{Var}(Y)+\big(\mathbb{E}[-Y]-\delta(r)\big)^2}
\\  &\ge
\frac{\mu(r)^2\,\tau}{\mu(r)^2\,\tau+\tfrac{2}{5}L_r^2}
=: \epsilon(r,\tau) >\ 0.
\end{align*}
This verifies \textbf{V2}(ii) uniformly for all $x\in K_r$.

This example highlights the advantage of analyzing almost sure reachability directly in continuous time using the proposed certificates, in contrast to the discrete-time Euler approximation, for which no polynomial drift certificate exists.

\section{Reachability Characterization for Linear SDEs}\label{sec:linear-sde}
Linear SDEs are a canonical benchmark for understanding almost sure reachability in continuous time. They arise as local linearizations of nonlinear systems, as well as fundamental models in control, physics, and finance~\cite{lindquist2015linear}. In this section, we specialize the drift--variant framework to linear SDEs with additive noise and show that almost sure reachability of a bounded target is determined by the spectral structure of the system matrices and, in the critical case, by the number of eigenvalues with zero real part.

\begin{theorem}\label{thm:cts-linear-ASR}
Consider SDE~\eqref{eq:sde} with $f(x(t))=Ax(t)$ and $g(x)=B\neq 0$, 
where $A \in \mathbb{R}^{n\times n}$, $B \in \mathbb{R}^{n\times m}$ are constant matrices and $B$ is full row rank. Let $G \subset \mathbb{R}^n$ be an open, bounded set containing the origin. Denote the spectral abscissa $\alpha(A) := \max\{\operatorname{Re}(\lambda) : \lambda \in \sigma(A)\}$  and let $E_A$ be the real invariant subspace spanned by the real and imaginary parts of all eigenvectors of $A$ associated with eigenvalues $\lambda$ with $\operatorname{Re}(\lambda)=0$. Then:
\begin{enumerate}[label=(\roman*), leftmargin=1.4em]
\item[\textnormal{(i)}] If $\alpha(A) < 0$ (i.e., $A$ is Hurwitz), then $G$ is almost surely reachable from every initial condition. Moreover, there exist quadratic drift and variant certificates $V(x)=x^\top P x$ and $U(x)=x^\top S x - b$ satisfying \textbf{V1}--\textbf{V2}.

\item[\textnormal{(ii)}] If either $\alpha(A) > 0$, or $\alpha(A) = 0$ but there exists a Jordan block of size greater than one associated with some eigenvalue $\lambda$ with $\operatorname{Re}(\lambda)=0$, then no norm-like drift $V$ can satisfy the drift condition~\textbf{V1} for any compact set $C$ and therefore $G$ is not almost sure reachable.

\item[\textnormal{(iii)}] Suppose $\alpha(A)=0$ and all eigenvalues with $\operatorname{Re}(\lambda)=0$ have Jordan blocks of size one. Then:
\begin{itemize}
\item If $\dim(E_A) \le 2$, then $G$ is almost sure reachable and  one can construct a logarithmic-type drift $V(x)= \sqrt{\ln\|x\|_\star}$ (for a suitable weighted norm $\|\cdot\|_\star$) and a quadratic variant $U(x)=\|x\|^2_\star - b$ satisfying \textbf{V1}--\textbf{V2}.

\item If $\dim(E_A) \ge 3$, then almost sure reachability of $G$ fails.
\end{itemize}
\end{enumerate}
\end{theorem}

\begin{pf} We provide a sketch of the proof. 
Part~(i) one can solve the Lyapunov equation $A^\top P + P A = -Q$ with $Q\succ 0$ and set $V(x)=x^\top P x$. Then
\begin{equation*}
\mathcal{A}V(x) = -x^\top Q x + \operatorname{tr}(P B B^\top)\le 0    
\end{equation*}
outside an ellipsoid, hence, \textbf{V1} holds. A quadratic variant $U(x)=x^\top S x - b$ satisfies \textbf{V2} by combining It\^o’s formula with the full rankness of $B$.

For part~(ii), if $\alpha(A)>0$ or there is a Jordan block of size greater than $1$ with $\operatorname{Re}(\lambda)\ge 0$, the deterministic flow $\mathrm{e}^{At}x_0$ grows exponentially or polynomially in some direction. One can show that for suitable $x_0$ the deterministic term dominates the stochastic integral with positive probability, so that $\|x(t)\|\to\infty$ on a set of positive probability. This violates the existence of a norm-like drift $V$ with $\mathcal{A}V\le 0$ outside a compact set, hence almost sure reachability fails.

For part~(iii), decompose $\mathbb{R}^n$ into the invariant subspace $E_A$ corresponding to eigenvalues with $\operatorname{Re}(\lambda)=0$ and its stable complement associated with $\operatorname{Re}(\lambda)<0$. On the stable part, the argument from (i) gives a quadratic Lyapunov function. On $E_A$, the restricted dynamics are a Brownian motion in dimension $k=\dim(E_A)$.

If $k\le 2$, this process is recurrent. One can construct a weighted norm $\|\cdot\|_\star$ and a logarithmic drift $V(x)=\sqrt{\ln\|x\|_\star}$ whose generator is non-positive outside a compact set, together with a quadratic variant $U(x)=\|x\|^2_\star-b$ that decreases over a time interval with positive probability, yielding \textbf{V1}--\textbf{V2} and hence almost sure reachability.

If $k\ge 3$, the projection of $x(t)$ onto $E_A$ behaves like a $k$-dimensional Brownian motion with bounded drift and full-rank $B$, which is known to be transient and hence contradicts the almost sure reachability. \myqed
\end{pf}

Therefore the almost sure reachability behavior of the continuous-time linear SDE is completely determined by the spectral structure of $A$ and, in the critical case, by the dimension of the subspace $E_A$. This is the continuous-time analogue of the discrete-time characterization for linear systems with additive noise in~\cite{kordabad2025certificates}.

\begin{remark}
Consider the constant-coefficient SDE $\mathrm{d}x(t)=A\mathrm{d}t + B\,\mathrm{d}W(t)$. If $A=0$, the process reduces to a pure diffusion evolving inside the subspace generated by the columns of $B$.  This corresponds to the linear setting of Theorem~\ref{thm:cts-linear-ASR} with $A=0$. Then for $n\leq 2$ the system almost surely reaches any bounded open target set containing zero; when $n\geq 3$ the system violates the almost sure reachability. If $A\neq 0$, the constant drift pushes the state in a fixed direction and grows linearly in time, eventually dominating the stochastic terms.   Consequently, the state escapes to infinity with positive probability and no drift certificate can exist and no bounded target set is almost surely reachable.
\end{remark}

In the next section, we move beyond linear dynamics and develop sum-of-squares based certificates synthesis procedures for polynomial stochastic systems, turning the drift and variant conditions into concrete semidefinite programs that can be solved numerically.

\section{Sum-of-Squares Certificates for Polynomial Systems}\label{sec:sos}
In this section, we focus on polynomial stochastic systems because the inequalities appearing in the drift--variant criteria can be encoded as sum-of-squares (SOS) constraints by an appropriate choice of the templates for the certificates. Then the constraints can be translated directly into semidefinite programs (SDPs) solvable by standard optimization tools. As in the discrete-time setting, SOS enables us to represent properties such as radial unboundedness, non-positive generator outside a compact set, and robust progress toward the target as SOS constraints once proper templates are fixed. Inspired by the double-well Langevin system, we fix a polynomial template for the drift function and choose the variant function as an exponential function composed with a polynomial. This allows us to translate conditions \textbf{V1--V2} into SOS constraints. Hence, it provides a solver-ready and numerically tractable framework for certifying almost sure reachability of continuous-time stochastic systems.

We emphasize that the SOS formulation in this section is a sound but not complete computational relaxation. The necessity result in Section~\ref{sec:cts-certificates} guarantees the existence of drift--variant certificates under almost sure reachability, but it does not imply that these certificates can always be chosen within the fixed polynomial/exponential-polynomial templates used here. Moreover, SOS constraints are themselves sufficient relaxations of polynomial nonnegativity. Therefore, failure of the SOS program to find certificates should not be interpreted as failure of almost sure reachability.

Throughout this section, we consider polynomial stochastic systems so that all relevant expressions in the generator $\mathcal{A}$ and the variance inequalities remain polynomial and can be handled with SOS methods.

\begin{Assumption}[Polynomial SDEs]\label{assum:poly-f-g}
Consider the SDE~\eqref{eq:sde} with $X=\mathbb{R}^n$. 
The functions $f$ and $g$ are assumed to be polynomial in the state. 
Specifically,
\begin{equation*}
f(x)=(f_1(x),\dots,f_n(x))^\top,\qquad f_i\in\mathbb{R}[x]\ \text{for all}\ i,
\end{equation*}
and each entry of $g(x)$ is polynomial, \begin{equation*}
g(x)=\big[g_{ij}(x)\big]_{i,j=1}^{n,m},\qquad g_{ij}\in\mathbb{R}[x]\ \text{for all}\ i,j,
\end{equation*}
where $\mathbb{R}[x]$ is the ring of polynomials in the variable $x$ with real coefficients. Moreover, the target set $G$ is given by polynomial inequalities
\begin{equation*}
G=\{x\in\mathbb{R}^n:g_i(x)<0,\ i=1,\dots,I\},\qquad g_i\in\mathbb{R}[x].
\end{equation*}
\end{Assumption}

\begin{Definition}
For $x \in \mathbb{R}^n$, a multivariate polynomial $p(x)\in \mathbb{R}[x]$ is an SOS, and denoted as $p(x)\in \Sigma[x]$, if there exist some polynomials $p_i(x)\in\mathbb{R}[x]$, $i = 1, \ldots, M$, such that
$
    p(x) = \sum_{i=1}^{M} p_i^2(x)
$.
\end{Definition}
An equivalent characterization of SOS polynomials has been provided by~\cite{parrilo2000structured}. Specifically, a  polynomial $p(x)$ of degree $2\kappa$ is SOS if and only if there exists $Q \succeq 0$ and a vector of monomials $z(x)$ containing all monomials in $x$ of degree $\leq \kappa$ such that
$
    p(x) = z(x)^\top Q z(x). $
Note that all SOS polynomials are nonnegative, but the converse is not necessarily true; that is, there exist nonnegative polynomials that are not SOS (see e.g.,~\cite{ahmadi2012convex}). However, this characterization enables SOS problems to be formulated as SDPs, which can be efficiently solved using convex optimization techniques.

\subsection{SOS-based Drift Function}
We seek a polynomial function $V\in\Sigma[x]$ satisfying~\textbf{V1}. To ensure $V$ is radially unbounded, we require
\begin{equation}\label{eq:ct-unbounded-V}
V(x)-\gamma_0\,x^\top x+\lambda_0\in\Sigma[x],\qquad \gamma_0>0,\ \lambda_0\in\mathbb{R}.
\end{equation}
This guarantees $V(x)\ge\gamma_0\|x\|^2-\lambda_0$, and hence $V(x)\to\infty$ as $\|x\|\to\infty$. 

Next, since $V$ is twice continuously differentiable, we use~\eqref{eq:inf:exp} and impose a non-positive generator condition outside a compact set through
\begin{align}\label{eq:ct-AV-neg}
&-\Big(\underbrace{f(x)^\top\nabla V(x)
+\tfrac12\operatorname{tr}\big(g(x)^\top\nabla^2V(x)g(x)\big)}_{=\mathcal{A}V(x)}\Big)\nonumber\\
&\qquad-\gamma_1\,x^\top x+\lambda_1\in\Sigma[x],
\end{align}
for some $\gamma_1,\lambda_1>0$. This implies $\mathcal{A}V(x)\leq -\gamma_1\,x^\top x+\lambda_1\leq \lambda_1$ for all $x\in\mathbb{R}^n$. Moreover, one can observe $\mathcal{A}V(x)\le0$ for all $x$ such that $x^\top x>{\lambda_1}/{\gamma_1}$. Therefore \textbf{V1} is satisfied with $d=\lambda_1$ and $C=\{x:x^\top x\le{\lambda_1}/{\gamma_1}\}$. 
The constraints~\eqref{eq:ct-unbounded-V}--\eqref{eq:ct-AV-neg} yield an SOS feasibility program whose solution provides a polynomial drift certificate $V$ for the continuous-time system.

\subsection{SOS-based Variant Function}
We now construct the variant certificate $U$ satisfying~\textbf{V2}. Motivated by the double-well example, we adopt a function obtained by composing an exponential function with a polynomial, given by
\begin{equation}\label{eq:U:poly}
U(x):=\frac{1}{\lambda}\big(1-\mathrm{e}^{-\lambda\,\zeta(x)}\big),\qquad \lambda>0,
\end{equation}
where $\zeta\in\mathbb{R}[x]$ is a polynomial to be designed. 
Since $s\mapsto(1-\mathrm{e}^{-\lambda s})/\lambda$ is monotonically increasing, we have $U(x)\le0$ if and only if $\zeta(x)\le0$. Note that the growth rate of 
$U(x)$ in~\eqref{eq:U:poly} decreases for large  $\zeta(x)$. Therefore, this suggests a suitable candidate for \textbf{V2}.

Similar to the example in Section~\ref{sec:exam:cont}, we have
\begin{equation*}
U(x) \le\ H(r):=\frac{1}{\lambda}\Big(1-\mathrm{e}^{-\lambda \zeta_{\max}(r)}\Big),
\quad \forall x\in C_V(r),    
\end{equation*} 
where $C_V(r):=\{x\in \mathbb{R}^n: V(x)\le r\}$ and $\zeta_{\max}(r):=\max_{x\in C_V(r)} \zeta(x)$. Note that $C_V(r)$ is bounded for the radially unbounded polynomial $V$ for each $r<\infty$.

We enforce $\{x:U(x)\le0\}\subseteq G$ by the following set of SOS constraints:
\begin{equation}\label{eq:ct-target-contain}
-g_i(x)+S_i(x)\zeta(x)-\alpha_i\in\Sigma[x],\,\, S_i\in\Sigma[x],\,\, i=1,\dots,I,
\end{equation}
for some $\alpha_i>0$ and all $i=1,\dots,I$. Applying the infinitesimal generator to $U$ yields
\begin{equation*}
\mathcal{A}U(x)=\mathrm{e}^{-\lambda\,\zeta(x)}\,\beta(x),
\end{equation*}
where
\begin{align*}
\beta(x):=&f(x)^\top\nabla\zeta(x)
+\tfrac12\operatorname{tr}\big(G(x)\nabla^2\zeta(x)\big)\nonumber\\ &
-\tfrac12\lambda\,\nabla\zeta(x)^\top G(x)\nabla\zeta(x),
\end{align*}
and $G(x):=g(x)g(x)^\top$. On any region where $\zeta(x)\ge0$, we have $\mathrm{e}^{-\lambda\zeta(x)}\in(0,1]$, hence $\mathcal{A}U(x)\le\beta(x)$. 
A uniform negative mean drift can then be enforced by
\begin{equation}\label{eq:ct-beta-neg}
-\beta(x)-\mu-\Lambda(x)\zeta(x)\in\Sigma[x],\, \Lambda\in\Sigma[x],
\end{equation}
with $\mu>0$. 
This ensures $\beta(x)\le-\mu$ whenever $\zeta(x)\ge0$. Since $U(x)$ is twice continuously differentiable and under Assumption~\ref{assum:poly-f-g}, we have
\begin{equation*}
    \lim_{t\downarrow 0}\frac{\mathbb{E}_x[U(x(t))]-U(x)}{t}
=\mathcal{A}U(x)=\mathrm{e}^{-\lambda\,\zeta(x)}\,\beta(x).
\end{equation*}
Therefore, for any $x$ such that $U(x)\ge0$, there exists a sufficiently small $\tau>0$  such that,
\begin{equation*}
    \frac{\mathbb{E}[U(x(\tau))-U(x)]}{\tau}\le -\frac{\mu}{2}.
\end{equation*}
By Itô’s isometry,
\begin{equation*}
\mathrm{Var}(U(x(\tau))-U(x))\le\mathbb{E}\int_0^\tau\|g(x(s))^\top\nabla U(x(s))\|^2\mathrm{d}s.
\end{equation*}
Since $\nabla U(x)=\mathrm{e}^{-\lambda\zeta(x)}\nabla\zeta(x)$, we have
\begin{equation*}
    \|g(x)^\top\nabla U(x)\|^2
=
\mathrm{e}^{-2\lambda\zeta(x)}\|g(x)^\top\nabla\zeta(x)\|^2.
\end{equation*}
Although this is not a polynomial function of $x$, it is continuous. Hence, it is bounded on every compact sublevel set $\{x:V(x)\le r\}$. Therefore, for every fixed $r<\infty$ and $\tau>0$, the random variable $U(x(\tau))-U(x)$ has finite variance for all initial conditions $x$ satisfying $V(x)\le r$. We denote by $\Gamma_{r,\tau}>0$ an upper bound on this variance, that is, $\operatorname{Var}\!\big(U(x(\tau))-U(x)\big)\le \Gamma_{r,\tau}$. Combining this with the mean inequality $\mathbb{E}[U(x(\tau))-U(x)]\le-\frac{\mu}{2}\tau$ and applying Cantelli’s inequality~\cite{ghosh2002probability} gives
\begin{equation*}
\mathbb{P}\big(U(x(\tau))-U(x)\!\le\!-\frac{\mu}{4}\tau\big)\!\ge\!
\frac{\mu^2\tau^2}{\mu^2\tau^2+16\Gamma_{r,\tau}}\!:=\! \bar \epsilon(r,\tau)\!>\!0,
\end{equation*}
 Hence, the variant condition~\textbf{V2} holds with $h(r)=\tau$, $\delta(r)=\tfrac14\mu\,\tau$, and $\epsilon(r)=\bar \epsilon(r,\tau)$.

The SOS constraints \eqref{eq:ct-target-contain},\eqref{eq:ct-beta-neg} contain bilinear couplings between $(\zeta(x),\lambda)$, $(\zeta(x),\Lambda(x))$ and $(\zeta(x), S_{1,\ldots, I}(x))$. 
In order to solve this, we introduce a slack variable $\varepsilon$ and maximize it subject to the SOS constraints, such that positivity margins are enforced as $\mu\ge\varepsilon$ and $\alpha_i\ge\varepsilon$. Concretely,
\begin{align}\label{eq:ct-variant-opt}
\max_{\substack{
\zeta(x), S_{1,\ldots,I}(x), \Lambda(x),  \\
\lambda, \mu, \alpha_{1,\ldots,I}, \varepsilon
}}\qquad &\varepsilon\\
\text{s.t.}  \qquad\qquad\quad &-\beta(x)-\mu-\Lambda(x)\zeta(x)\in\Sigma[x],\nonumber\\ & \Lambda\in\Sigma[x],\,\,
\lambda\ge\varepsilon,\,\, \mu\ge\varepsilon,\nonumber\\
 &\forall i=1,\dots,I:\nonumber
\\&\,\,g_i(x)+S_i(x)\zeta(x)-\alpha_i\in\Sigma[x],\nonumber\\ &\,\,S_i\in\Sigma[x],\,\,
\alpha_i\ge\varepsilon.\nonumber
\end{align}
The role of $\varepsilon$ is twofold: it acts as a common positive margin for target containment and mean-decrease constraints, and its maximization drives the certificates away from degeneracy; feasibility with $\varepsilon>0$ yields a valid variant with quantitative margins $\alpha_i\ge\varepsilon$ and $\mu\ge\varepsilon$. Using the following standard practice in SOS synthesis, we then resolve these via an \emph{alternating scheme}. We first fix $\zeta(x)$ and optimize multipliers $(\lambda, \Lambda(x), S_{1,\ldots,I}(x))$, margins $(\mu, \alpha_{1,\ldots,I})$ and slack variable $\varepsilon$, then fix multipliers $(\lambda, \Lambda(x), S_{1,\ldots,I}(x))$ and update $\zeta(x)$, margins $(\mu, \alpha_{1,\ldots,I})$ and slack variable $\varepsilon$, repeating until $\varepsilon>0$ is obtained. Hence, each optimization turns into a linear SOS program.

\section{Simulation Results}\label{sec:simulations}
In this section, we present two numerical examples, one linear and one polynomial, to illustrate the implications of our reachability results and the SOS-based certificate synthesis.

\subsection{A Simple Linear System}
\label{subsec:linear}

To demonstrate Theorem~\ref{thm:cts-linear-ASR} without constructing certificates, we consider a simple linear system $\mathrm dx (t) = \sigma \mathrm dW (t)$, i.e., $A=0$ and $B=\sigma I$ in~\eqref{eq:sde}. We fix a small open target set $G:=\{x\in\mathbb R^n:\|x\|<1\}$ containing the origin, initialize the process at a point with $\|x_0\|>1$, and estimate the cumulative distribution function of the hitting time $\mathbb P(\tau_G \le t)$ for $n\in\{1,2,3,4\}$ using $100$ Monte Carlo trajectories for each dimension $n$. Figure~\ref{fig:0} shows a dimension-dependent behavior. In dimensions $n=1$ and $n=2$, the empirical probability of entering $G$ increases rapidly and approaches $1$ as the horizon grows, whereas in dimensions $n=3$ and $n=4$ it saturates strictly below $1$ over the same range of horizons. This matches Theorem~\ref{thm:cts-linear-ASR}.

\begin{figure}[t]
  \centering
  \includegraphics[width=0.47\textwidth]{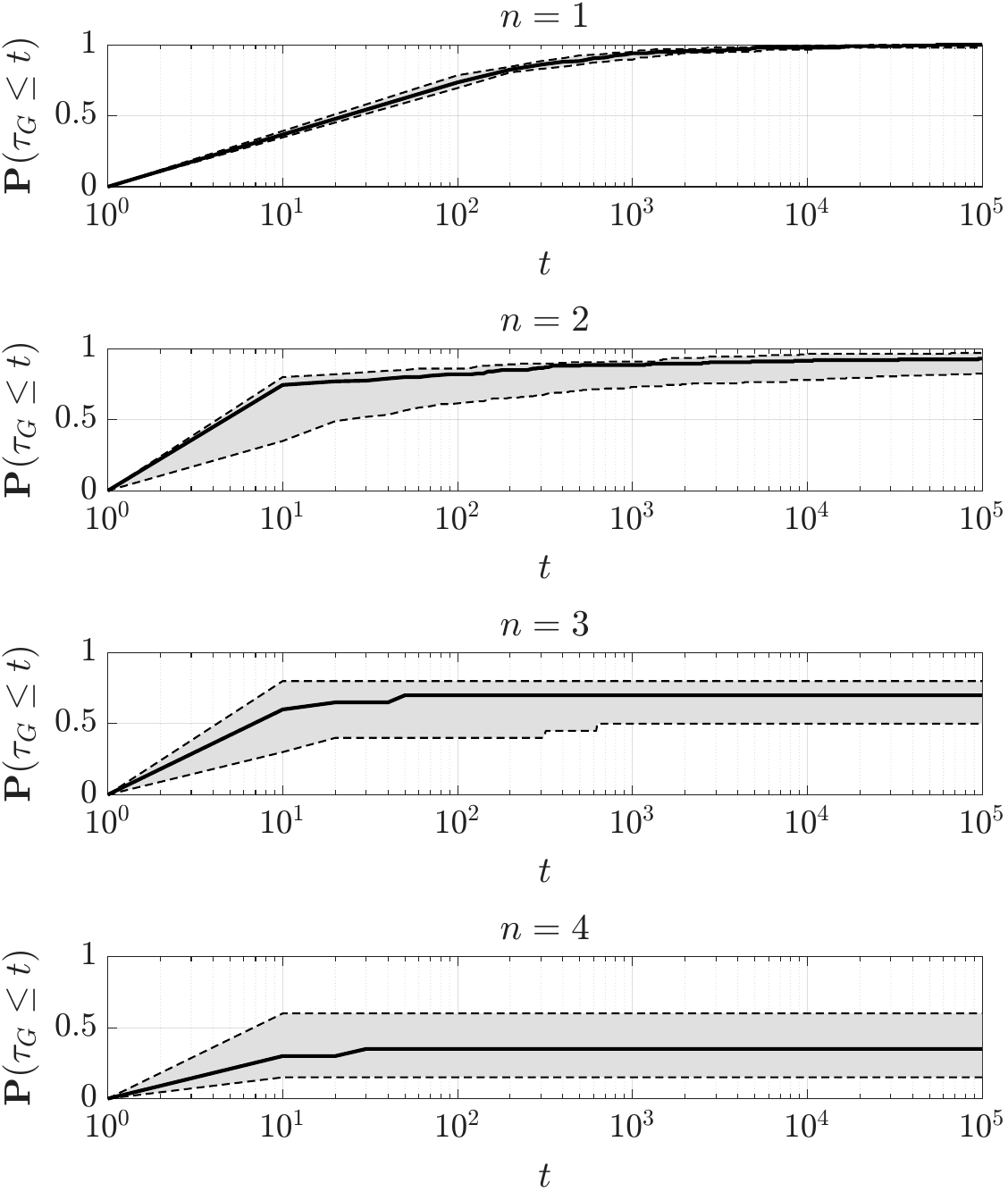}
  \caption{Empirical cumulative distribution function of the hitting time $\mathbb P(\tau_G\le t)$ for $n\in\{1,2,3,4\}$. The curves for $n=1,2$ approach $1$ with increasing horizon, while the curve for $n=3,4$ saturates below $1$. The shaded regions indicate the $10$–$90$\% percentile range across simulations.}
  \label{fig:0}
\end{figure}

\subsection{Wolfe--Quapp Potential}\label{subsec:wq_example}

We consider a two--dimensional stochastic system evolving on the Wolfe--Quapp potential energy surface, which is a standard benchmark for metastability and rare transitions between multiple basins~\cite{zhang2013double,aguilar2010implementation}. The potential energy surface $\mathcal U:\mathbb{R}^2\to\mathbb{R}$ is defined by
\begin{align*}
\mathcal U(x_1(t),x_2(t))
=&x_1(t)^4 + x_2(t)^4
- 2x_1(t)^2- 4x_2(t)^2+\\ & 
 x_1(t)x_2(t)
+ 0.3x_1(t) + 0.1x_2(t).
\end{align*}
The corresponding overdamped Langevin dynamics (a diffusion on this potential) is
\begin{equation}\label{eq:sde:WQ}
\mathrm dx(t)
=
-\nabla \mathcal U(x(t))\,\mathrm dt+\sqrt{2\theta^{-1}}\,\mathrm dW(t),
\end{equation}
where $x(t) := [x_1(t), x_2(t)]^\top$ and $\theta=1.4$ is the inverse temperature, and $W(t)$ is standard
two-dimensional Brownian motion and
\begin{equation*}
\nabla \mathcal U(x_1(t),x_2(t))
=
\begin{bmatrix}
4x_1(t)^3 - 4x_1(t) + x_2(t) + 0.3\\[1ex]
4x_2(t)^3 - 8x_2(t) + x_1(t) + 0.1
\end{bmatrix}.
\end{equation*}
Figure~\ref{fig:1} shows the level
sets of the Wolfe--Quapp potential and overlaid sample paths generated using Euler--Maruyama discretization. The target set
boundary $\partial G$ is shown in black and two trajectories are shown in blue and red. It can be seen that both trajectories enter the target
set. Hence, the noise induces transitions between basins, and trajectories eventually hit $G$. The blue trajectory enters the target set at time $t\approx556.6$.

Figure~\ref{fig:2} shows the synthesized radially
unbounded drift certificate using SOS and  Figure~\ref{fig:25} shows the generator value $\mathcal{A}V(x)$ together with its
zero contour. This plot highlights that $\mathcal{A}V(x)$ becomes non-positive outside a compact set while being uniformly upper bounded, which is
the continuous--time drift condition as in \textbf{V1}.

Finally, Figure~\ref{fig:4} shows the one-step decrease probability $\mathbb{P}(U(x(\tau)) - U(x) \le -\delta)$ for $\delta=0.001$ and $\tau=0.01$ 
over the region $\{x: U(x)>0\}$ for the synthesized variant function that is uniformly positive. This figure confirms condition~\textbf{V2} and together with Figures~\ref{fig:2} and~\ref{fig:25} visualize the two certificates establishing almost sure reachability for the Wolfe--Quapp SDE~\eqref{eq:sde:WQ}.
\begin{figure}[t!]
    \centering
    \includegraphics[width=0.4\textwidth]{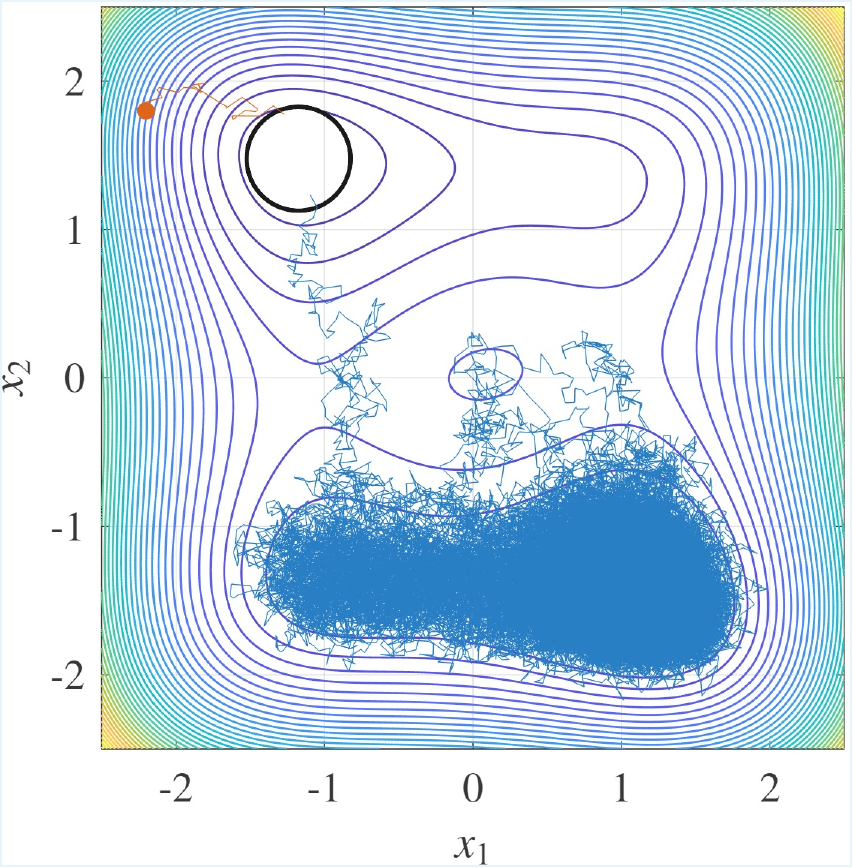}
    \caption{Wolfe--Quapp potential level sets with sample paths of the overdamped Langevin dynamics. The target set $G$ is shown by its boundary $\partial G$ in black, and two trajectories (in blue and red) are simulated until they hit $G$.}
    \label{fig:1}
\end{figure}

\begin{figure}[t!]
    \centering
    \includegraphics[width=0.4\textwidth]{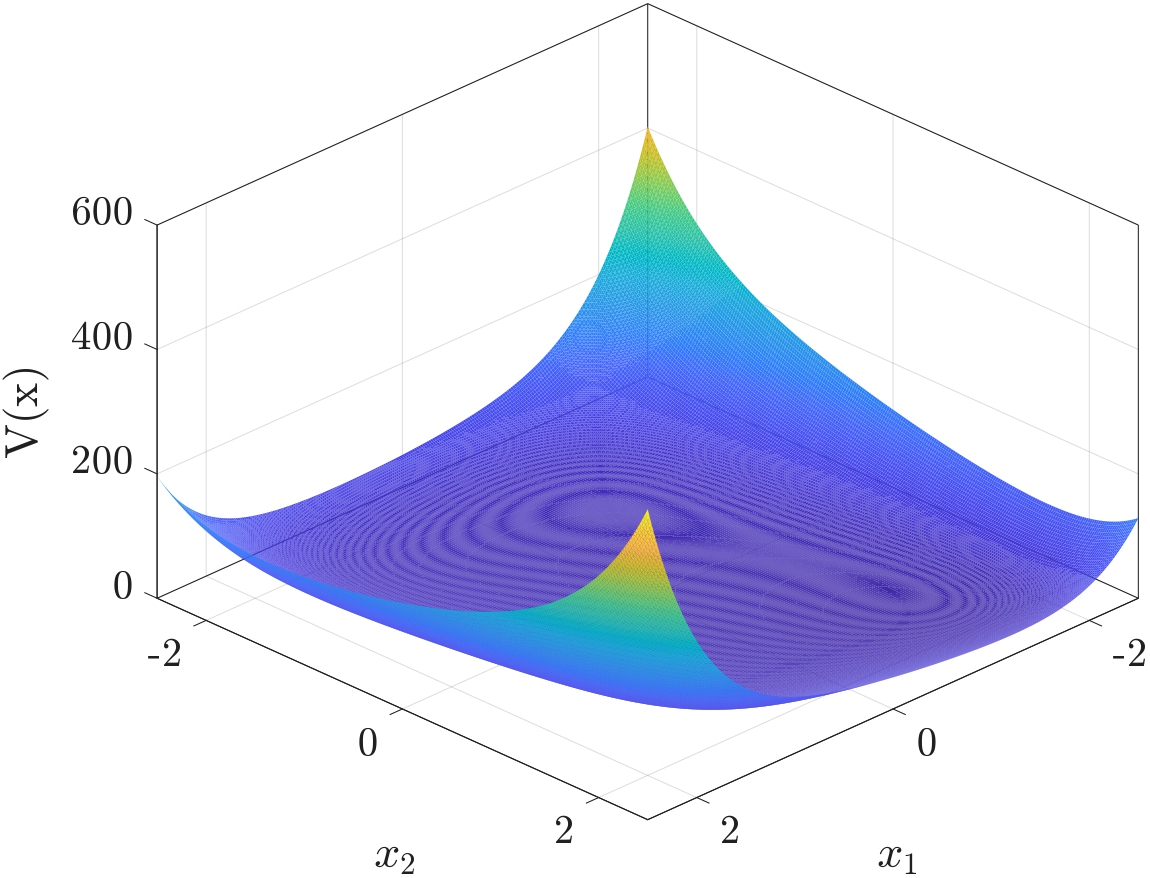}
    \caption{SOS-based Synthesized radially unbounded drift certificate for the Wolfe--Quapp SDE~\eqref{eq:sde:WQ}.}
    \label{fig:2}
\end{figure}

\begin{figure}[t!]
    \centering
    \includegraphics[width=0.4\textwidth]{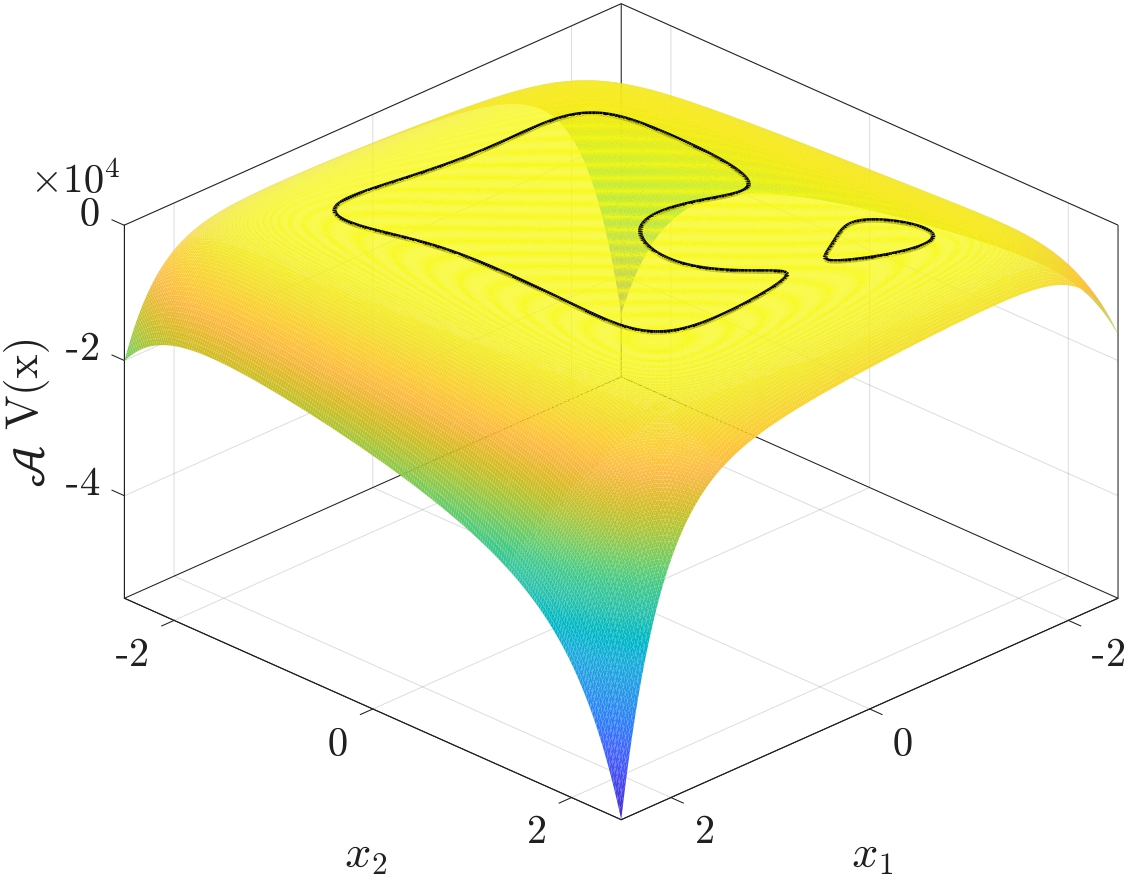}
    \caption{Generator $\mathcal{A}V(x)$ with its zero contour. The non-positivity of $\mathcal{A}V(x)$ outside a compact region visualizes the drift condition that prevents escape to infinity. It can be seen that the generator is also uniformly upper bounded.}
    \label{fig:25}
\end{figure}

\begin{figure}[t!]
    \centering
    \includegraphics[width=0.4\textwidth]{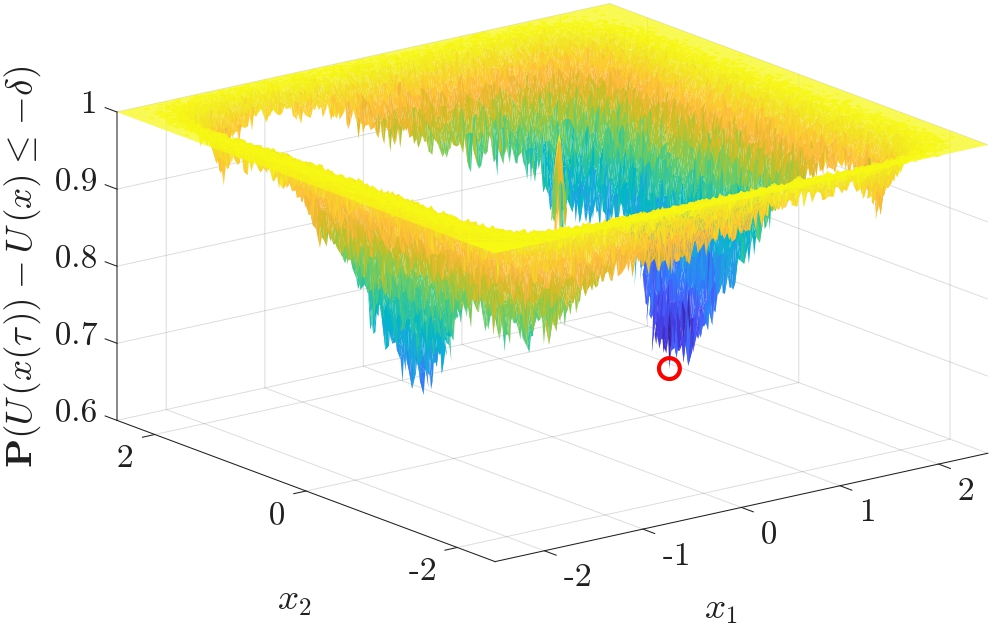}
    \caption{One-step decrease probability 
$\mathbb{P}(U(x(\tau)) - U(x) \le -\delta)$ 
over the region $\{U(x)>0\}$. The red marker indicates the minimum probability over $\{x: U(x)>0\}$. }
    \label{fig:4}
\end{figure}

\section{Conclusion}\label{sec:conclusion}
We developed a certificate-based framework for almost sure reachability in continuous-time stochastic systems. Starting from a double-well Langevin example, we showed that Euler-type discretizations may break almost sure reachability and rule out polynomial drift functions, even when the original SDE is almost sure reachable. This motivates working directly in continuous time. We then introduced drift and variant certificates, and established that together they are sufficient and necessary for almost sure reachability of an open precompact target under mild regularity and weak Feller assumptions.

Specializing to linear SDEs with additive noise, we derived a complete classification of almost sure reachability based on the spectral properties of the drift matrix and the way noise excites the neutral subspace. For polynomial SDEs, we showed how to encode the drift--variant conditions as sum-of-squares constraints, leading to semidefinite programs that synthesize polynomial drift and variant certificates via an alternating optimization scheme.

Future work includes extending the framework to controlled SDEs and certificates for safety specifications and quantitative properties.

 \bibliographystyle{elsarticle-num}
\bibliography{Cont_Arxiv}
\end{document}